\documentclass[aoas2]{imsart}

\arxiv{1404.7530}

\usepackage{amsfonts,amsmath,amssymb,amsthm}
\usepackage[round,authoryear]{natbib}

\usepackage{verbatim}
\usepackage{url}
\usepackage{array}
\usepackage{booktabs}
\usepackage{afterpage}
\usepackage[dvipsnames]{color}
\usepackage{mfirstuc}
\usepackage{times}
\makeatletter
\newif\if@restonecol
\makeatother

\usepackage[lined,algonl,boxed]{algorithm2e}
\usepackage{epsfig}
\usepackage{multirow}
\usepackage{rotating}

\usepackage[hidelinks]{hyperref}

\usepackage[caption=false,labelformat=simple,font={sf,sl},hangindent=0pt,format=hang,singlelinecheck=false,position=top,topadjust=-20pt]{subfig}

\newcommand{\Ord}{\ensuremath{\mathcal{O}}}

\newtheorem{theorem}{Theorem}[section]

\newtheorem{corollary}[theorem]{Corollary}

\newtheorem*{rep@theorem}{\rep@title}
\newcommand{\newreptheorem}[2]{%
\newenvironment{rep#1}[1]{%
 \def\rep@title{#2 \ref{##1}}%
 \begin{rep@theorem}}%
 {\end{rep@theorem}}}

\newreptheorem{theorem}{Theorem}

\renewcommand\vert{\, | \,}
\renewcommand\Pr{\mathrm{Pr}}
\newcommand\E{\mathrm{E}}

\newcommand\am{\mathbf{A}}
\newcommand\bm{\mathbf{B}}
\newcommand\xm{\mathbf{X}}
\newcommand\dm{\mathbf{D}}

\begin{document}

\begin{frontmatter}
\title{Design and analysis of experiments in networks: Reducing bias from interference}
\runtitle{Design and analysis of experiments in networks}

\begin{aug}
\author{\fnms{Dean} \snm{Eckles}}
\author{\fnms{Brian} \snm{Karrer}}
\and
\author{\fnms{Johan} \snm{Ugander}\thanksref{t1}}

\affiliation{Facebook, Inc., Facebook, Inc., and Cornell University}

\runauthor{Eckles, Karrer \& Ugander}

\end{aug}




\begin{abstract}
Estimating the effects of interventions in networks is complicated when
the units are interacting, such that the outcomes for one unit may depend on
the treatment assignment and behavior of many or all other units (i.e., there is interference).
When most or all units are in a single connected component, 
it is impossible to directly experimentally compare outcomes
under two or more global treatment assignments
since the network can only be observed under a single assignment.
Familiar formalism, experimental designs, and analysis methods assume the
absence of these interactions, and result in biased estimators of causal effects of interest.
While some assumptions can lead to unbiased estimators, these assumptions are generally unrealistic,
and we focus this work on realistic assumptions.
Thus, in this work, we evaluate methods
for designing and analyzing randomized experiments
that aim to reduce this bias and thereby reduce overall error.
In \emph{design}, we consider the ability to perform random assignment
to treatments that is correlated in the network, such as through graph cluster
randomization.
In \emph{analysis}, we consider incorporating information about the
treatment assignment of network neighbors.
We prove sufficient conditions for bias reduction through both design and analysis
in the presence of potentially global interference.
Through simulations of the entire process of experimentation in networks,
we measure the performance of these methods under varied network structure
and varied social behaviors, 
finding substantial bias and error reductions.
These improvements are largest for networks with more clustering and data generating processes
with both stronger direct effects of the treatment and stronger interactions between units.
\end{abstract}

\thankstext{t1}{Authors are listed alphabetically.}

\end{frontmatter}

\section{Introduction}

\noindent
Many situations and processes of interest to scientists involve individuals
interacting with each other, such that causes of the behavior of one individual are 
also indirect causes of the behaviors of other individuals; that is, there are \emph{peer effects}
or \emph{social interactions} \citep{manski_economic_2000}.
Likewise, in applied work, the policies considered by decision-makers often have 
many of their effects through the interactions of individuals. 
Examples of such cases are abundant. In online social networks, the behavior of 
a single user explicitly and by design affects the experiences of other users in the network. 
If an experimental treatment changes a user's behavior, then it is reasonable 
to expect that this will have some effect on their friends, 
a perhaps smaller effect on their friends of friends, and so on out through the network.  
In an extreme case, treating one individual could alter the behavior of everyone in the network.

To see the challenges this introduces, consider what is, in many cases, a primary quantity of interest for experiments in networks --- the average treatment effect~(ATE) of applying a treatment to all units compared with applying a different (control) treatment to all units.\footnote{For example, \citet{bond_massive_2012} consider the effect of a voter mobilization intervention, such that the aim is to compare voter turnout if everyone (or almost everyone) is assigned to the treatment with turnout if everyone is assigned to the control.
There are other causal quantities that may be of interest, which we do not treat here.
Other authors consider decompositions of effects into various direct and indirect effects of the treatment \citep{sobel_randomized_2006,tchetgen_causal_2012,toulis_estimation_2013}.
}
Let $Z$ be a vector of length $N$ giving each unit's treatment assignment, so that $Y_i(Z = z)$ is the potential outcome of interest for unit $i$ when $Z$ is set to $z$. Then the ATE is a contrast between two such treatment vectors,
\begin{equation}
\label{tau_is_ate}
\tau(z_1, z_0) = \frac{1}{N} \sum_i \E[Y_i(Z = z_1) - Y_i(Z = z_0)],
\end{equation}
where $N$ is the number of units and $z_1$ and $z_0$ are two treatment assignments vectors;
the prototypical case has $z_1 = 1$ and $z_0 = 0$, the vectors of all ones and of all zeros.
Note that each unit's potential outcome is a function of the global treatment assignment vector $Z$,
not just its own treatment $Z_i$.
Additional assumptions will thus be required for $\tau$ to be identifiable.\footnote{This is closely connected to what \citet{holland_causal_1988} regards as the fundamental problem of causal inference --- that one can only observe a unit's response under a single treatment. The difference is that here we can only observe \emph{all} units' responses under a single global treatment.
}

The standard approach is to assume that each unit's response is not affected by the treatment
of any other units.
Versions of this assumption are sometimes called
the \emph{stable unit treatment value assumption} \citep[SUTVA; ][]{rubin_estimating_1974}
or a \emph{no interference} \citep{cox_planning_1958} assumption.
Combined with random assignment to treatment, this suffices to identify $\tau$.
However, for many processes and situations of interest the units are interacting,
and SUTVA becomes implausible \citep{aronow_estimating_2011,sobel_randomized_2006}.

Rather than substituting other strong assumptions about interference,
this paper considers how we can reduce bias for the ATE through
both the choice of experimental design and analysis
when interactions among units occur along an observed network.\footnote{While we limit the analysis here to cases where the measured network and the network through which the interaction occur are the same, the methods examined here may also substantially reduce bias in when using a network observed with error.
}
The design of the experiment 
dictates how each vertex in the network (i.e., unit) is assigned to a condition, 
and the analysis says how the observed responses are combined into 
estimates of causal quantities of interest.
We study these methods by 
formalizing the process of experimentation in networks,
proving sufficient conditions for bias reduction through design and analysis,
and running extensive simulations.

We cannot consider all possible designs and analysis,
but limit this work to some relatively general methods for each.
We consider experimental designs that assign clusters of vertices
to the same treatment; this is \emph{graph cluster randomization} \citep{ugander_graph_2013}.
Since the counterfactual situations of interest involve all
vertices being in the same condition, the intuition is that assigning a vertex and 
vertices near it in the network into the same condition, the vertex is 
``closer'' to the counterfactual 
situation of interest.
For analysis methods, we consider methods that define \emph{effective treatments} such that
only units that are effectively in global treatment or global control are used to estimate the ATE.
For example, an estimator for the ATE might only compare units in treatment that are surrounded
by units in treatment with units in control that are surrounded by units in control.
The intuition is 
that a unit that meets one of these conditions is ``closer'' to a 
counterfactual situation of interest.

The rest of the paper is structured as follows.
We briefly review some related work on experiments in networks.
Section \ref{sec:model} presents a model of the process of experimentation in networks,
including initialization of the network, treatment assignment, outcome generation,
and analysis.
This formalization allows us to develop theorems giving sufficient conditions for bias reduction.
To develop further understanding of the magnitude of the bias and error reduction in practice,
Section \ref{sec:simulations} presents simulations using networks generated from
small-world models and then degree-corrected blockmodels.

We find that graph cluster randomization is capable of dramatically reducing bias
compared to independent assignment without adding ``too much'' variance.
The benefits of graph cluster randomization are larger when the network has more local clustering and
when social interactions are strong. 
If social interactions are weak or the network has little local clustering,
then the benefits of the more complex graph-clustered design are reduced.
Finally, we found larger bias and error reductions through design than analysis:
analysis strategies using neighborhood-based definitions of effective treatments does 
further reduce bias,
but often at a substantial cost to precision such that the simple estimators were
preferable in terms of error.
No combination of design and analysis is expected to work well across very different situations,
but these general insights from simulation can be a guide to practical real-world experimentation in the presence of peer effects.
Furthermore, by identifying sufficient conditions for bias reduction, we can understand when
design and analysis changes will at least not increase bias.

\subsection{Related work}
Much of the literature on interference between units focuses on situations where there
are multiple independent groups, such that there are interactions within, but not between, groups
\citep[e.g.,][]{sobel_randomized_2006,rosenbaum_interference_2007, hudgens_toward_2008, tchetgen_causal_2012}.
Some more recent work has examined interference in networks more generally
\citep{aronow_estimating_2011, manski_identification_2013, toulis_estimation_2013, ugander_graph_2013},
where this between-groups independence structure cannot be assumed.

This prior work has largely focused on assuming restrictions on the extent
of interference (e.g., vertices are only affected by the number of neighbors treated)
and then deriving results for designs and estimators motivated by these same assumptions.
\citet{aronow_estimating_2011} give unbiased estimators for ATEs under these assumptions
and derive variance estimators.\footnote{\citet{aronow_estimating_2011} also consider estimating the effects of peer assignment on ego behavior directly, while our primary quantity of interests is the ATE of global treatment versus global control.}
\citet{ugander_graph_2013} show that graph clustered randomization puts more vertices in the
conditions required for these estimators, such that the variance of these estimators is bounded for
certain types of networks.
But, as noted by \citet{manski_identification_2013} and 
as we discuss in Section \ref{sec:implausibility} below,
the very processes expected to produce interference also make
these assumptions implausible.
The present work explicitly considers more realistic data generating processes that 
violate these restrictive assumptions.
That is, in contrast to prior work, we evaluate design and analysis strategies
under conditions other than those under which they have particular desirable
properties (e.g., unbiasedness).
Instead, we settle for reducing bias and error.\footnote{In this regard, the present work is more similar to \citet{toulis_estimation_2013},
which recognizes that available estimators of the quantities of interest will be biased.
}

\section{Model of experiments in networks}
\label{sec:model}
We consider experimentation in networks as consisting of four phases: (i) {\it initialization}, (ii) {\it treatment assignment}, (iii) {\it outcome generation}, and (iv) {\it estimation}. A single run through these phases corresponds to a single instance of the experimental process. Treatment assignment embodies the experimental design, and the estimation phase embodies the analysis of the network experiment. These same phases, shown in Figure~\ref{diagram}, are implemented in our simulations in which we instantiate this process many times.

\begin{figure}[t]
\begin{center}
\includegraphics[width=\textwidth]{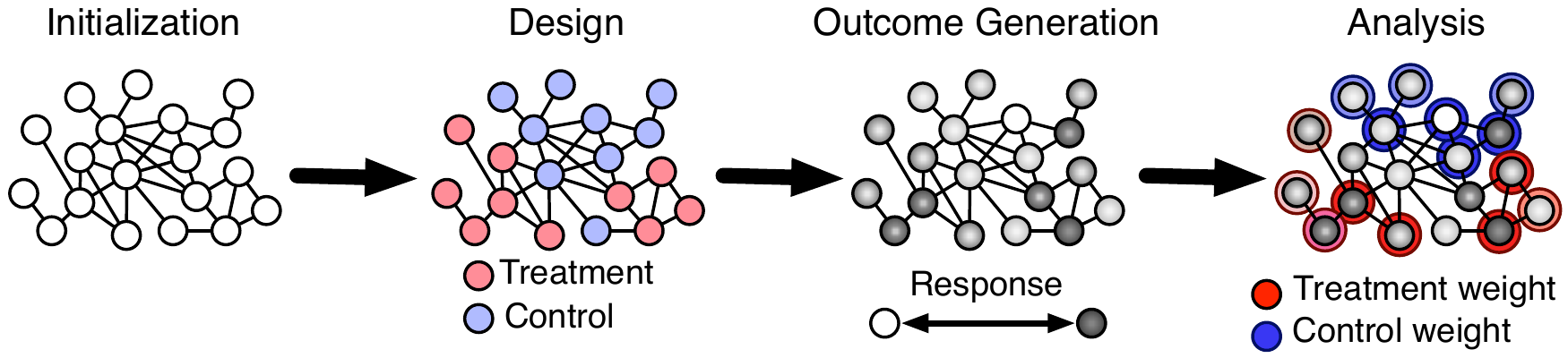}
\caption{Model of the network experimentation process, consisting of 
\emph{(i) initialization}, which generates the graph and vertex characteristics, 
\emph{(ii) design}, which determines the randomization scheme,
\emph{(iii) outcome generation}, which observes or simulates behavior, and
\emph{(iv) analysis}, which constructs an estimator. We examine 
the bias and variance of treatment effect estimators under
different design and analysis methods for varied initialization and outcome
generation processes.}
\label{diagram}
\end{center}
\end{figure}

\subsection{Initialization}
Initialization is everything that occurs prior to the experiment. This includes network formation and the processes that produce vertex characteristics and prior behaviors. In some cases, we may regard this initialization process as random, and so wish to understand design and analysis decisions averaged over instances of this process; for example, we may wish to average over a distribution of networks that corresponds to a particular network formation model. 
In the simulations later in this paper, we generate networks from small-world models \citep{watts_collective_1998}
and degree-corrected blockmodels \citep{karrer_stochastic_2011}.
In other cases, we may regard the outcome of this process as fixed; for example, we may be working with a particular network and vertices with particular characteristics, which we wish to condition on in planning our design and analysis.

When initialization is complete, we have a particular network $G=(V,E)$ with adjacency matrix $\am$.\footnote{For the purposes of this paper, we assume that the network is fixed over the timescale of the experiment.} 
In addition to producing a graph, the initialization process could also produce a collection of vertex characteristics $\xm$ that may or may not relate to the structure of the graph, but may play a role in outcome generation.

\subsection{Design: Treatment assignment}
The treatment assignment phase creates a mapping from vertices to treatment conditions. We only consider a binary treatment here (i.e., an ``A/B'' test), so the mapping is from vertex to treatment or control. Treatment assignment normally involves independent assignment of units to treatments, such that one unit's assignment is uncorrelated with other units' assignments.\footnote{A normal but minor exception occurs when forcing a specific number of units within a block to be assigned to each of treatment and control; this produces negative dependence between units in the same block. This includes global balancing of sample sizes in treatment and control as a special case.
} In this case, each unit's treatment is a Bernoulli random variable
$$
Z_i \sim \text{Bernoulli}(q)
$$
with probability of assignment to the treatment $q$.

The present work evaluates treatment assignment procedures that
produce assignments with network autocorrelation. 
While many methods could produce such network autocorrelation, 
we work with graph cluster randomization, 
in which the network is partitioned into clusters and those clusters are used to assign treatments.
Let the vertices be partitioned into $N_C$ clusters $C_1, C_2, ..., C_{N_C}$, and define $C(\cdot): \{1, ..., N\} \rightarrow \{1, ..., N_C\}$ as mapping vertex indices to cluster indices. Thus $C_i$ refers to a cluster by its index, while $C(i)$ refers to the cluster containing vertex $i$. 

In standard graph cluster randomization, as presented by \citet{ugander_graph_2013}, treatments are assigned at the cluster level, where each cluster $C_j$ is assigned a treatment $W_j \sim \text{Bernoulli}(q)$. Thus the treatments assigned to vertices are simply those assigned to their clusters,
$$
Z_i = W_{C(i)}.
$$
For some estimands and analyses, assigning all vertices in a cluster to the same treatment can make it impossible for some vertices to be observed with, e.g., some particular number of treated peers. This can violate the requirement that all units have positive probability of assignment to all conditions. For this reason, it can be desirable
to use an assignment method that allows for some vertices to be assigned to a different treatment than the rest of its cluster; we describe such a modification in Appendix \ref{appendix:hole_punching}.

Graph cluster randomization could be applied to any mapping $C(\cdot)$ of vertices to clusters. One such mapping, which we use for the simulations reported in this paper, is formed by $\epsilon$-net clustering as previously considered by \citet{ugander_graph_2013}.  An $\epsilon$-net in the graph distance metric is a set of vertices such that no two vertices in the set are less than $\epsilon$ hops of each other, and every vertex outside the set is within $\epsilon$ hops (in fact, $\epsilon-1$ hops) of a vertex in the set.  An $\epsilon$-net can be formed by repeatedly selecting a vertex and removing it and every vertex within distance $\epsilon-1$ from the network, until all vertices have been removed.   Having completed this step, the population of selected vertices forms an $\epsilon$-net.  An $\epsilon$-net clustering can be formed by assigning each vertex to the closest vertex in the $\epsilon$-net, and breaking the possible ties through some arbitrary rule.  Different selection and assignment rules and different values of $\epsilon$ correspond to different experimental designs.  We compare clustered random assignment using $\epsilon$-nets to independent random assignment, where vertices are independently assigned to treatment and control.

Other mappings of vertices to clusters of interest include methods developed for community detection \citep{fortunato_community_2010}. 
Many global community detection methods, such as modularity maximization \citep{newman_modularity_2006},
have a resolution limit such that they do not distinguish small clusters \citep{fortunato_resolution_2007}; 
graph cluster randomization with these methods could then introduce
too large an increase in variance for the resulting bias reduction.
Therefore, local clustering methods may be more appealing for graph cluster randomization.  
Observed community membership (e.g., current educational institution) or geography could also be used as this mapping.

Lastly, it is important to note that independent random assignment can be considered as clustered random assignment where each vertex is in its own cluster.

\subsection{Outcome generation and observation}
\label{observed_outcomes}
Given the network (along with vertex characteristics and prior behavior) 
and treatment assignments, some data generating process 
produces the observed outcomes of interest.
In the context of social networks,
typically this is the unknown process by which individuals make their decisions.
In this work, we consider a variety of such processes.
For our simulations, we use a known process meant to simulate decisions,
in which units respond to others' prior behaviors.
Doing so allows us to understand the performance of varied design and analysis methods, 
measured in terms of estimators' bias and error, 
under varied (although simple) decision mechanisms.
Before considering these processes themselves,
we consider outcomes as a function of treatment assignment.

\subsubsection{Treatment response assumptions}
\label{sec:treatment_response_assumptions}
In the following presentation, we use the language of 
``treatment response'' assumptions developed by \citet{manski_identification_2013}
to organize our discussion of outcome generation. 
Consider vertices' outcomes as determined by a function from the 
global treatment assignment $Z \in \mathbb{Z}^N$ 
and an independent stochastic component
 $U \in \mathbb{U}^N$ to an outcome vector $Y \in \mathbb{Y}^N$:
$$
f(\cdot): \mathbb{Z}^N \times \mathbb{U}^N \rightarrow \mathbb{Y}^N.
$$
We then observe $Y = f(Z, U)$. We can decompose this function into a function for each vertex
$$
f_i(\cdot): \mathbb{Z}^N \times \mathbb{U}^N \rightarrow \mathbb{Y}.
$$
We can, as we have done above, continue to write $Y_i(Z = z)$ to refer to the outcome for vertex $i$ that would be observed under assignment $z$; by suppressing dependence on $U$, this treats $Y_i(\cdot)$ as a stochastic function.

If vertices' outcomes are not affected by others' treatment assignment, then SUTVA is true.
Perhaps more felicitously, \citet{manski_identification_2013} calls this assumption \emph{individualistic treatment response} (ITR).
Under ITR we could then consider vertices as having a function from only their own assignment to their outcome:
$$
f_i(\cdot): \mathbb{Z} \times \mathbb{U}^N \rightarrow \mathbb{Y}.
$$
One way for this assumption to hold is if the vertices do not interact.\footnote{The vertices might interact without necessarily violating the ITR assumption.
This can occur, for example, when vertices interact in one period, and then are affected by treatment assignment, while no longer interacting.
This is why we define $f_i(\cdot)$ as being a function from $\mathbb{U}^N$ rather than just $\mathbb{U}$.
}
This specification of $f_i(\cdot)$ corresponds to the assumption that a vertex's outcome is invariant to changes in other vertices' assignments. That is, for any two global assignments $z_0, z_1 \in \mathbb{Z}^N$ and any stochastic component $U \in \mathbb{U}^N$,
$$
z_{1,i} = z_{0,i} \Rightarrow f_i(z_{1}, U) = f_i(z_{0}, U).
$$
ITR is a particular version of the more general notion of \emph{constant treatment response} (CTR) assumptions \citep{manski_identification_2013}. 
More generally, a CTR assumption involves establishing equivalence classes of treatment vectors by defining a function
$g_i(\cdot): \mathbb{Z}^N \rightarrow \mathbb{G}_i$ that maps global treatment vectors to the space $\mathbb{G}_i$ of \emph{effective treatments} for vertex $i$ \citep{manski_identification_2013}  such that
$$
g_i(z_1) = g_i(z_0) \Rightarrow f_i(z_{1}, U) = f_i(z_{0}, U)
$$
for any two global assignments $z_0, z_1 \in \mathbb{Z}^N$ and any stochastic component $U \in \mathbb{U}^N$.
Specifying the functions $g_i$ is then a general way to specify a CTR assumption. Such assumptions can be described as constituting an \emph{exposure model} \citep{aronow_estimating_2011, ugander_graph_2013}. 

Other CTR assumptions have been proposed that allow for some interference. \citet{aronow_estimating_2011} simply posit different restrictions on this function, such as that a vertex's outcome only depends on its assignment and its neighbors' assignments. This \emph{neighborhood treatment response} (NTR) assumption has that, for any two global assignments $z_0, z_1 \in \mathbb{Z}^N$ and any stochastic component $U \in \mathbb{U}^N$,
$$
z_{1,i} = z_{0,i} \text{ and } z_{1,\delta(i)} = z_{0,\delta(i)} \Rightarrow f_i(z_{1}, U) = f_i(z_{0}, U),
$$
where $\delta(i)$ are the neighbors of vertex $i$. \citet{aronow_estimating_2011} and \citet{ugander_graph_2013} consider further restrictions, such as that a vertex's response only depends on the number of treated neighbors.

\subsubsection{Implausibility of tractable treatment response assumptions}
\label{sec:implausibility}
How should we select an exposure model? \citet[Section 3]{aronow_estimating_2011} suggest that we ``must use substantive judgment to fix a model somewhere between the traditional randomized experiment and arbitrary exposure models''. However, it is unclear how substantive judgement can directly inform the selection of an exposure model for experiments in networks --- at least when the vast majority of vertices are in a single connected component.
Interference is often expected because of peer effects: in discrete time, then the behavior of a vertex at $t$ is affected by the behavior of its neighbors at $t - 1$; if this is the case, then the behavior of a vertex at $t$ would  also be affected by the behavior of its neighbors' neighbors at $t - 2$, and so forth. Such a process will result in violations of the NTR assumption, and many other assumptions that would make analysis tractable. \citet{manski_identification_2013} shows how some, quite specific, models of simultaneous endogenous choice can produce some restrictions on $f_i(\cdot)$.\footnote{\citet{manski_identification_2013} calls these models of simultaneous endogenous choice a ``system of structural equations''. But because these equations are simultaneous, they are not structural in the sense of corresponding to a directed acyclic graph (DAG) given a causal interpretation \citep{pearl_causality:_2009}. However, we can regard these equations as specifying an equilibrium that arises out of some unknown dynamic process. We prefer to work with a posited dynamic process, which may or may not be in equilibrium when we observe it \citep[cf.][]{young_individual_1998}.}

Since many appealing CTR assumptions are violated by the very theories that motivate the expectation of interference, it is useful to evaluate the performance of available design and analysis methods --- including estimators that would be motivated by these assumptions ---
under outcome generating processes consistent with these theories.
In particular, we now consider outcome generating processes in which vertices respond to their own treatment and the prior behavior of their neighbors. That is, peer behavior fully mediates the effects of the assignments of an ego's peers on the ego. This is notably different from \citet{aronow_estimating_2011} and \citet{ugander_graph_2013}, where ego response is specified in terms of peer assignments without being mediated through peer behavior.\footnote{
Note that this specification in terms of ``direct'' effect could be compatible with various data-generating processes that involve ``indirect'' effects --- at least on a short time scale.
}

We consider a dynamical model with discrete time steps in which a vertex's behavior at time $t$, denoted by the vector $Y_{i,t}$, is a function $h$ of ego treatment assignment and it and its neighbors' prior behaviors $Y_{\delta{(i)},t-1}$, such that
$$
h_{i,t} (\cdot): \mathbb{Z} \times \mathbb{Y}^{k_i + 1} \times \mathbb{U}^N \rightarrow \mathbb{Y},
$$
where $k_i$ is the degree of vertex $i$ and $Y_{\cdot,0}$ is initialized by some prior process. That is, $h_{i,t}(\cdot)$ is the nonparametric structural equation (NPSE) for $Y_{i,t}$.

Together with the graph $G$, the function $h_{i,t}(\cdot)$ determine the treatment response function $f_i(\cdot)$. Thus, this outcome generating process implies some CTR assumptions. After the first time step (i.e., at time 1), the effective treatment for a vertex, the function $g_i(\cdot)$ considered earlier, maps to the space of the vertex's treatment. After the second time step, it maps to the space of the vertex's treatment and its neighbors treatment. After the third time step (i.e., at time 3), the effective treatment is no finer than the treatment of all vertices within distance 2.  At time step $t$, the effective treatment is no finer than the treatment of all vertices within distance $t - 1$. We see here that under such a dynamic outcome generating process, Manski's notion of effective treatment, conceived of to limit the scope of dependence, quickly expands to encompass the full graph.\footnote{And similarly for the assumed exposure models in \citet{aronow_estimating_2011} and \cite{ugander_graph_2013}.}

\subsubsection{Utility linear-in-means}
\label{section:utility_linear-in-means}
Many familiar models are included in the above outcome generating process. To make this more concrete, and for our subsequent simulations, we consider a model in which a vertex's behavior is a stochastic function of the mean of neighbors' prior behaviors, so that behavior at some new time step $t$ is generated as:
\begin{align}
\label{probit_model}
Y_{i, t}^* &= \alpha +  \beta Z_i+ \gamma \frac{{A}_i^{'} Y_{t - 1}}{k_i}  + U_{i,t} \\ 
\label{probit_model_link}
Y_{i, t} &= a\big( Y_{i, t}^* \big)
\end{align}
where $A_i$ is a row of the adjacency matrix and $k_i$ is the degree of vertex $i$. In the case of a binary behavior, we work with $a(x) = 1\{x > 0\}$ and $U_{i,t} \sim \mathcal{N}(0, 1)$, which is a probit model. We initialize behaviors with $Y_{i, 0} = 0$. Here $\alpha$ is the baseline, where a negative $\alpha$ determines the threshold that must be crossed for $Y_{i, t}^*$ to be positive. Setting $\beta$ determines the strength of the direct effect of the treatment, while $\gamma$ is the slope for peer behavior, and therefore determines the strength of the peer effects.  This process is then run up to a maximum time $T$.  As described above, with a small value of $T$, this implies CTR assumptions.

This can be interpreted as a \emph{noisy best response} or \emph{best reply} model \citep{blume_statistical_1995},
when vertices anticipate neighbors taking the same action in the present round as they did in the previous round.
In particular, we can interpret $Y_{i, t}^*$ as the payoff for vertex $i$ to adopt behavior 1 at time $t$.
When $\gamma > 0$, then this is a semi-anonymous graphical game with strategic complements \citep[Ch. 9]{jackson_social_2008}.

\subsection{Analysis and estimation}
\label{estimation}
We focus on the ATE (the average treatment effect; $\tau$ in Equation \ref{tau_is_ate}), which is naturally of interest when considering whether a new treatment would be beneficial if applied to all units.

There are many options available for estimating the ATE. For example, if the relevant network is completely unknown or if peer effects are not expected, then one might use estimators for experiments without interference, such as a simple difference-in-means between the outcomes of vertices assigned to treatment and control. To clarify the sources of error in estimation, we begin with the population analogs of these quantities --- i.e., the associated estimands --- and return to the estimators themselves in Section \ref{sec:estimators}.
Consider the simple difference-in-means estimand
\begin{equation}
\label{eq:tau_itr}
\tau^d_{\mathrm{ITR}}(1, 0) = \mu^d_\mathrm{ITR}(1) - \mu^d_\mathrm{ITR}(0)
\end{equation}
where the $\mu^d_\mathrm{ITR}$ are mean outcomes when a vertex is in treatment and control, i.e.,
$$
\mu^d_\mathrm{ITR}(z) = \frac{1}{N} \sum_{i=1}^N \E^d[Y_i \vert Z_i = z_i].
$$
We index these quantities by both the definition of effective treatments (ITR for ``individualistic treatment response'', as in Section \ref{sec:treatment_response_assumptions}) and the experimental design $d$, since the former determines the conditioning involved and the latter determines the distribution of $Z$ over which we take expectations.

When a vertex's outcome depends on the treatment assignments of others, these quantities need not equal the quantities of interest. That is, they can suffer from some estimand bias, such that 
$
\tau^d_{\mathrm{ITR}}(1, 0) - \tau(1, 0)
$
is non-zero. Each vertex assigned to treatment contributes to this bias through the difference between its expected outcome when assigned to treatment (given the experimental design) and what would be observed under global treatment.
More generally, for some global treatment vector $z$, vertex $i$ contributes to the bias of $\mu^d_\mathrm{ITR}(z)$ through
$
\E^d[Y_i - Y_i(Z = z) \vert Z_i = z_i].
$
If the treatment assignment of other vertices do not affect vertex $i$'s behavior much, 
then this contribution might be quite small. Or this contribution could be more substantial.

\subsubsection{Bias reduction through design}
We are now equipped to elaborate on the intuition that graph cluster randomization puts vertices in conditions ``closer'' to 
the global treatments of interest and thereby reduces bias in estimates of average treatment effects, 
even if a vertex's outcome depends on the global treatment vector. The result below uses a linear outcome model that has as a special case the linear-in-means model, as made clear at the end of this subsection. 

\begin{theorem}
\label{prop:linear_bias_reduction_design}
Assume we have a linear outcome model for all vertices $i \in V$ such that
\begin{eqnarray}
\E_U[Y_i(z, U)] = a_i + \sum_{j \in V} B_{ij} z_j
\end{eqnarray}
and further assume that $Y_i(z, u)$ is monotonically increasing in $z$
for every $u \in \mathbb{U}^N$ and 
vertex $i$
such that $B_{ij} \geq 0$.

Then for some mapping of vertices to clusters,
the absolute bias of $\tau^{d}_\mathrm{ITR}(1,0)$ when the design $d$ is graph cluster randomization is less than or equal to 
the absolute bias when $d$ is independent assignment, with a fixed treatment probability $p$.
\end{theorem}

\begin{proof}
Using the linear model for $Y_i$ and the definition of $\tau$, we have that the true ATE $\tau$ is given by
\begin{eqnarray}
\label{eq:linear_outcome_model}
\tau(1,0) = \mu(1) - \mu(0) = \frac{1}{N} \sum_{ij} B_{ij}
\end{eqnarray}
for this outcome model.  Under graph cluster randomization,
\begin{eqnarray}
\tau^\mathrm{gcr}_\mathrm{ITR}(1,0) = \frac{1}{N} \sum_{ij} B_{ij} \mathbf{1}[C(i) = C(j)].
\end{eqnarray}
Then under independent assignment,
\begin{eqnarray}
\tau^\mathrm{ind}_\mathrm{ITR}(1,0) = \frac{1}{N} \sum_{i} B_{ii}.
\end{eqnarray}
Because $B_{ij} \geq 0$, together this implies that
$\tau(1,0) - \tau^\mathrm{gcr}_\mathrm{ITR}(1,0)  \leq 
\tau(1,0)-\tau^\mathrm{ind}_\mathrm{ITR}(1,0)$, where monotonicity dictates that each side of this inequality is positive. 
\end{proof}

This comparison allows seeing how, at least in this linear model, the magnitude of bias reduction from graph cluster randomization depends on the ``strength'' of the interactions within clusters. That is, this clarifies the intuition that using clusters formed from more distant vertices will not generally reduce bias as much as clusters formed from closer vertices, as is the aim of using graph partitioning methods such as $\epsilon$-net partitioning or community detection methods.\footnote{Note that in the above treatment, the mapping of vertices to clusters is not random, so any mapping is bias reducing. 
}
It also highlights that when there are mainly non-zero $B_{ij}$'s, \emph{ceteris paribus} large clusters result in more bias reduction; of course, there are corresponding costs to precision.

To clarify this further, let's consider the relative bias defined by
\begin{eqnarray}
\tau^\mathrm{gcr}_\mathrm{ITR}(1,0)/\tau(1,0) - 1 = \frac{\sum_{ij} B_{ij} \mathbf{1}[C(i) = C(j)]}{\sum_{ij} B_{ij}} - 1.
\end{eqnarray}
Assume that there are $\Ord(N)$ clusters of size $\Ord(1)$ used for the graph cluster randomization.\footnote{As shown by \citet{ugander_graph_2013}, assuming NTR and that the graph satisfies a restricted growth condition, this implies that 
an experimental design with $\Ord(N)$ clusters of size $\Ord(1)$ will produce NTR-based estimators with bounded variance.}
Under this condition, the numerator has $\Ord(N)$ terms and the denominator has $\Ord(N^2)$ terms.  So unless there is a judicious choice of clustering, the numerator will be overwhelmed by the denominator and the estimator $\tau^\mathrm{gcr}_\mathrm{ITR}(1,0)$ will be a dramatic underestimate of the true average treatment effect, and it's clear that $\tau^\mathrm{ind}_\mathrm{ITR}(1,0)$ would be even worse.  In order for meaningful relative bias reduction to occur, the clustering must capture the structure of the dependence between units specified by the matrix of coefficients $\bm$.

In Appendix \ref{appendix:balanced_linear_case}, we derive similar intuitions from an alternative graph cluster randomization that preserves balance between the sizes of the treatment and control group.  There graph cluster randomization no longer always achieves bias reduction for every clustering over independent assignment, but meaningful bias reduction is again possible and depends on how the clustering captures $\bm$ in an identical way.

This linear outcome model has as special cases some other models of interest. In particular, it has as a special case the linear-in-means model,
which is widely studied and used in econometrics \citep[e.g.,][]{manski_identification_1993_reflection,lee_identification_2007,bramoulle_identification_2009,goldsmith_social_2013}.
Consider $a(x) = x$ in Eq.~\ref{probit_model}.
Then for $t \geq 1$ the quantity $E^U[Y_{i,t}(z)]$ is
\begin{eqnarray}
\E^U[Y_{i,t}(z)] = \alpha +  \beta z_i+ \gamma \frac{{A}_i^{'} \E^U[Y_{t - 1}(z)]}{k_i}.
\end{eqnarray}
The closed form solution for $E^U[Y_{t}(z)]$ for any $t \ge 0$ is then given by
\begin{eqnarray}
\E^U[Y_{t}(z)] = (\gamma \dm^{-1} \am)^t \E^U[Y_0] + \sum_{q=0}^{t-1} (\gamma  \dm^{-1} \am)^q (\alpha + \beta z)
\end{eqnarray}
where $\dm^{-1}$ is the diagonal matrix of inverse degrees, $\am$ is the adjacency matrix, and $Y_0$ is the
vector of initial states.
This is a linear outcome model with 
$a_i = \alpha (1 - \gamma^t) / (1 - \gamma) + ((\gamma \dm^{-1} \am)^t \E^U[Y_0])_i$
and 
$B_{ij} = \beta \sum_{q = 0}^{t-1} (\gamma \dm^{-1} \am)^q_{ij}$.

\subsubsection{Bias reduction through analysis}
Definitions of effective treatments other than ITR correspond to different estimands.
In particular, we can incorporate assumptions about effective treatments into Equation \ref{tau_is_ate}.   Let
\begin{align}
\label{mu_g}
\mu^d_g(z) = \frac{1}{N} \sum_i \E^d[Y_i \vert g_i(Z) = g_i(z)]
\end{align}
be the mean outcome for the global treatment $z$ when $g$ specifies the effective treatments and $d$ is the experimental design.
Then we have
\begin{align}
\tau^d_g(z_1, z_0) &=  \mu^d_g(z_1) - \mu^d_g(z_0)
\label{effective_treatment_tau}
\end{align}
as our revised estimand for the ATE.\footnote{It is precisely the effective treatment assumption that allows generalization from a single sampled $z$ to the behavior at $z_1$ and $z_0$.}

If the effective treatment assumption corresponding to this estimator is satisfied, then it is unbiased.
As with the ITR assumption, we can again describe the bias that occurs when effective treatments are incorrectly specified. 
For some global treatment vector $z$, vertex $i$ contributes to the bias of $\mu^d_{g}(z)$ through
\begin{align}
\label{eq:bias_in_effective_contribution}
\E^d[Y_i - Y_i(Z = z) \vert g_i(Z) = g_i(z)],
\end{align}
where $g_i(\cdot)$ is the potentially incorrect (i.e., too coarse) specification of effective treatments for vertex $i$.

Considering two or more specifications of effective treatments can allow us to elaborate on the intuition that
using a finer specification of effective treatments will reduce bias by comparing only vertices that are in conditions ``closer'' to the global treatments of interest. For example, the NTR assumption corresponds to finer effective treatments than the ITR assumption.
We also relax the NTR assumption to a fractional $\lambda$-neighborhood treatment in which a vertex is considered effectively in global treated if a fraction $\lambda$ of its neighbors are treated (and the same for control) \citep{ugander_graph_2013}.

Here we analyze functions $g_i(\cdot)$ such that $g_i(Z) = g_i(z)$ just implies that 
for some subset of vertices $J_i$ we have that 
$\sum_{j \in J_i} 1\{Z_j = z_j\} \geq l_i$ and that $Z_i = z_i$. 
These are conditions such that some subset of size $l_i$ of a set of vertices $J_i$
has treatment assignment matching that in $z$, the global treatment vector of interest.
The fractional neighborhood treatment response (FNTR) assumption corresponds to such a function
with $J_i = \delta(i)$ and $l_i = \lceil \lambda k_i \rceil$, where $k_i$ is vertex $i$'s degree.
This has both ITR and NTR as special cases with $\lambda = 0$ and $\lambda = 1$ respectively.\footnote{Of course, ITR can also be analyzed with any choice of $J_i$, including the empty set.}

If we have two such functions $g_i^A(\cdot)$ and $g_i^B(\cdot)$ with the same $J_i$,
and $g_i^A(z) = g_i^A(z')$ implies $g_i^B(z) = g_i^B(z')$,
then we say that $g_i^A(\cdot)$ is \emph{more restrictive} than $g_i^B(\cdot)$.
  
\begin{theorem}
\label{prop:bias_reduction_analysis}
Let $g^A(\cdot)$ and $g^B(\cdot)$ be vectors of such functions where $g_i^A(\cdot)$ is more restrictive than $g_i^B(\cdot)$ for every vertex $i$, and let independent random assignment be the experimental design.  A sufficient condition for estimand $\tau^\mathrm{ind}_{g^A}(1, 0)$ to have less than or equal absolute bias than $\tau^\mathrm{ind}_{g^B}(1, 0)$,
where these estimands are defined by Equation \ref{effective_treatment_tau}, is that we have monotonically increasing responses or 
monotonically decreasing responses for every vertex with respect to $z$.
\end{theorem}

\begin{proof} Given in Appendix \ref{appendix:theory_analysis}.
\end{proof}
Note that the utility linear-in-means model in Equation \ref{probit_model} satisfies this monotonicity condition if the direct effect $\beta$ and peer effect $\gamma$ are both non-negative.

What about the combination of graph cluster randomization with these neighborhood-based estimands?
As we show in Appendix \ref{appendix:theory_analysis}, 
similar arguments apply if we count up matching \emph{clusters} instead of vertices,
but use of the FTNR estimand with graph cluster randomization is not necessarily bias reducing
under monotonic responses without this modification.

\subsubsection{Estimators}
\label{sec:estimators}
We now briefly discuss estimators for the estimands considered above.
First, we can estimate $\tau^d_\mathrm{ITR}(1, 0)$ with the difference in sample means
$
\hat{\tau}_{\text{I,S}}(1, 0) = \hat{\mu}_\text{I,S}(1) - \hat{\mu}_\text{I,S}(0)
$
where the $\hat{\mu}_\text{I,S}$ are simple sample means, i.e.,
$$
\hat{\mu}_\text{ITR,S}(z) = \frac{1}{\sum_{i=1}^N \mathbf{1}[Z_i = z_i]} \sum_{i=1}^N Y_i \mathbf{1}[Z_i = z_i].
$$
Note that these estimators are again indexed by the effective treatment $I$ used (i.e., ITR), but, unlike the estimands, they are not indexed by the design, though the design determines their distribution. We additionally distinguish these estimators by the weighting used (discussed below), identifying the simple (i.e., unweighted) means with~$\mathrm{S}$. If a vertex's own treatment is ignorable (as it is under random assignment, independent or graph clustered), then this estimator will be unbiased for $\tau^d_\mathrm{ITR}(1, 0)$.

More generally, there is a natural correspondence between the conditioning on $g_i(Z) = g_i(z)$ in the estimands and the vertices whose outcomes are used in an estimator. Given some specification of effective treatments $g$, one could construct an estimator of the ATE as a simple difference in the sample means for vertices in effective treatment and in effective control
$$
\hat{\tau}_{g,\mathrm{S}}(1, 0) = \hat{\mu}_{g,\mathrm{S}}(1) - \hat{\mu}_{g,\mathrm{S}}(0)
$$
where we have
$$
\hat{\mu}_{g,\mathrm{S}}(z) = \frac{
\sum_{i=1}^N Y_i \mathbf{1}[g_i(Z) = g_i(z)]
}{
{\sum_{i=1}^N \mathbf{1}[g_i(Z) = g_i(z)]} 
}.
$$
This estimator will only be unbiased for the corresponding estimand $\mu^d_{g}(z)$ under certain conditions.
To have an unbiased estimate of $\mu^d_g(z)$ using the sample mean requires that
$\E^d[Y_i \vert g_i(Z)=g_i(z)]$
be independent of
$\Pr^d[g_i(Z) = g_i(z)]$, the probability vertex $i$ is assigned to that effective treatment.
That is, the effective treatments must be ignorable.
One way for the effective treatments to be ignorable is if either of these quantities is the same for all vertices.
Usually we would not want to assume that $\E^d[Y_i \vert g_i(Z)=g_i(z)]$ is homogeneous, and 
$\Pr[g_i(Z) = g_i(z)]$ will not be homogeneous under many relevant effective treatments, such as neighborhood treatment response (NTR), since the distribution of effective treatments for a vertex depends on network structure.
As \citet{ugander_graph_2013} observe, high degree vertices will generally have low probability of being assigned to some kinds of ``extreme'' effective treatments, such as having all neighbors treated, while low degree vertices have a much higher probability of being in such an effective treatment.

Observed effective treatments can be made ignorable by conditioning on the design \citep{aronow_estimating_2011} or sufficient information about the vertices. The experimental design determines the probability of assignment to an effective treatment $\pi_i(z) = \Pr(g_i(Z) = g_i(z))$. In the case of graph cluster randomization and effective treatments determined by thresholds, these probabilities can be computed exactly using a dynamic program \citep{ugander_graph_2013}. These are generalized propensity scores that can then be used in Horvitz--Thompson estimators or other inverse-probability weighted estimators, such as the Hajek estimator \citep{aronow_estimating_2011} of the ATE. The Horvitz--Thompson estimator will often suffer from excessive variance, so we focus on the Hajek estimator:
\begin{align}
\hat{\tau}_{g,H}(z_1, z_0) = \; \; 
&\left(\sum_{i=1}^N \frac{\mathbf{1}[g_i(Z) = g_i(z_1)]}{\pi_i(z_1)}\right)^{-1} 
\sum_{i=1}^N \frac{Y_i \mathbf{1}[g_i(Z) = g_i(z_1)]}{\pi_i(z_1)} \; \; 
-  \nonumber\\ 
&\left(\sum_{i=1}^N \frac{\mathbf{1}[g_i(Z) = g_i(z_0)]}{\pi_i(z_0)}\right)^{-1} 
\sum_{i=1}^N \frac{Y_i \mathbf{1}[g_i(Z) = g_i(z_0)]}{\pi_i(z_0)} \; \; 
\label{hajek_estimate}
\end{align}
This estimator provides a nearly unbiased estimate of Equation \ref{effective_treatment_tau}.\footnote{The bias of the Hajek estimator is not zero, but it is typically small and worth the variance reduction. See \citet{aronow_estimating_2011}.}

Beyond bias, we also care about the variance of the estimator as well.
Estimators making use only of vertices with all neighbors in the same condition
will suffer from substantially increased variance, both because few vertices will be assigned to this effective treatment and because the weights in the Hajek estimator will be highly imbalanced.
This could motivate borrowing information from other vertices, such as by using additional modeling
or, more simply, through relaxing the definition of effective treatment, such as by using the fractional relaxation of the NTR assumption (FNTR).

The most appropriate effective treatment assumption to use for the analysis of a given experiment is not clear \emph{a priori}.
We will consider estimators motivated by two different effective treatments in our simulations.

\section{Simulations}
\label{sec:simulations}

In order to evaluate both design and analysis choices, we conduct simulations that instantiate the model of network experiments presented above. First, graph cluster randomization puts more vertices into positions where their neighbors (and neighbors' neighbors) have the same treatment; this is expected to produce observed outcomes ``closer'' to those that would be observed under global treatment. Second, estimators using fractional neighborhood treatment restrict attention to vertices that are ``closer'' to being in a situation of global treatment. Third, weighting using design-based propensity scores adjusts for bias resulting from associations between propensity of being in an effective treatment of interest and potential outcomes. Each of these three changes to design and analysis is expected to reduce bias, potentially at a cost to precision.
Under some conditions, we have shown above that these design and analysis methods reduce (or at least do not increase) bias for the ATE.
The goal of these simulations then is to characterize the magnitude of this bias reduction, weigh it against increases in variance,
and do so specifically under circumstances that do not meet the given sufficient conditions.

For each run of the simulation, we do the following. First, we construct a small world network with $N = \,$ 1,000 vertices and initial degree parameter $k=10$. We vary the rewiring probability $p_\mathrm{rw} \in \{0.00, 0.01, 0.10, 0.50, 1.00\}$, thereby producing both regular powers of the cycle ($p_\mathrm{rw} = 0$), graphs with ``small world'' characteristics ($p_\mathrm{rw} \in \{0.01, 0.10\}$), graphs with many random edges and less clustering ($p_\mathrm{rw} = 0.50$), and graphs with all random edges ($p_\mathrm{rw} = 1.00$).  The small world model of networks \citep{watts_collective_1998} is notable for being able to succinctly introduce clustering into an otherwise complex distribution over random graphs, all featuring a small diameter. The clustering of the graph, typically measured by the clustering coefficient, is a measure of the extent to which adjacent vertices share many common neighbors in the graph, and many social networks, including online social networks \citep[e.g.,][]{ugander_anatomy_2011}, have been found to exhibit a high degree of clustering as well as a small diameter.

For graph cluster randomization, we use a 3-net clustering and randomly assign each cluster in its entirety to treatment or control with equal probability.\footnote{Simulations for the Louvain method \citep{blondel_fast_2008} for community detection, not reported here, are qualitatively similar to those for $\epsilon$-net clustering, but generally resulted in more bias reduction but also larger variance increases, as expected by this method's resolution limit.}
We compare clustered assignment to independent random assignment.

We generate the observed outcomes using the probit model in Equations \ref{probit_model} and \ref{probit_model_link}, and set the baseline as $\alpha = -1.5$, making the behavior somewhat rare:
\begin{align}
Y_{i, t}^* = -1.5 +  \beta Z_i+ \gamma \frac{{A}_i^{'} Y_{i, t - 1}}{k_i}  + U_{i,t}, 
\hspace{0.5in}
Y_{i, t} = 1\{Y_{i, t}^* > 0\}.
\end{align}
We initialize $Y_{i,0} = 0$ for all vertices, and then run the process  for all combinations of $\beta \in \{0.0, 0.25, 0.5, 0.75, 1.0\}$ and $\gamma \in \{0.0, 0.25, 0.5, 0.75, 1.0\}$, up to a maximum time $T = 3$.\footnote{We also repeated these simulations with the small-world networks for $T = 10$. The results were qualitatively similar.}
Note that this data generating process does not
satisfy the conditions for graph cluster randomization to be bias reducing
given by Theorem \ref{prop:linear_bias_reduction_design}, since the outcome model is not linear.

Finally, for each simulation, we compute three estimates of the ATE.
The individual unweighted estimator (or difference-in-means estimator) 
$\hat{\tau}_{\mathrm{ITR,S}}$
makes no use of neighborhood information.
This is the baseline to which we compare the neighborhood unweighted estimator 
$\hat{\tau}_{\text{FNTR,S}}$
and the neighborhood Hajek estimator 
$\hat{\tau}_{\text{FNTR,H}}$, 
both using a fractional neighborhood treatment response (FTNR)
specification of effective treatments with $\lambda = 0.75$.
That is, these estimators count a vertex as being in effective treatment or effective control 
if at least three-fourths of its neighbors have the same assignment.
With independent assignment, the conditions for bias reduction given in Theorem \ref{prop:bias_reduction_analysis}
from using this estimator are satisfied.
With graph cluster randomization, it is not immediately obvious whether these conditions are satisfied (it may depend on details of the network).

\begin{figure}[tb]
\begin{center}
\includegraphics[width=.95\textwidth]{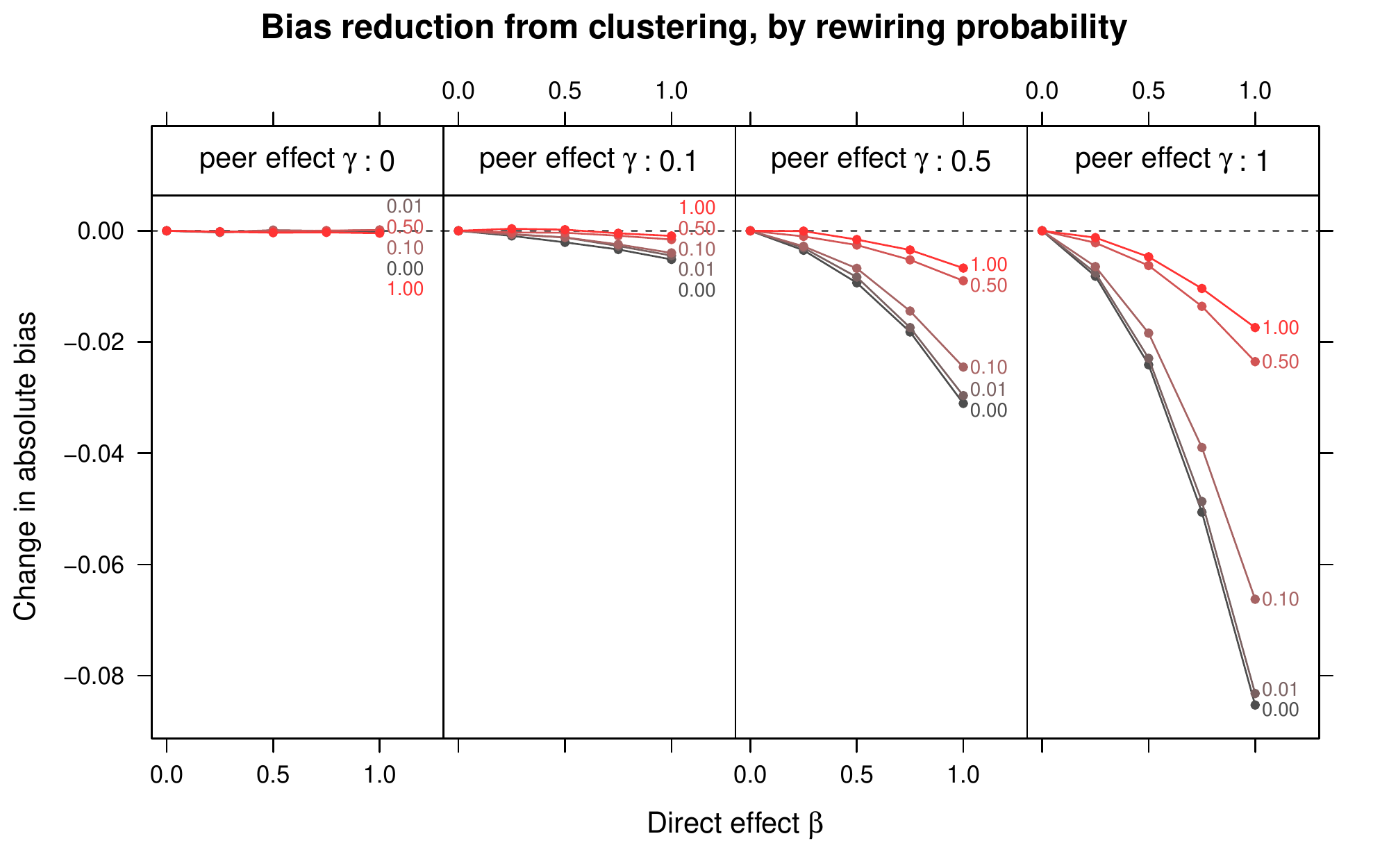}
\caption{Change in bias due to clustered random assignment as a function of 
the direct effect of the treatment $\beta$,
the rewiring probability $p_\mathrm{rw}$ (different colors),
and the strength of the peer effect $\gamma$ (different panels).
Random assignment clustered in the network reduces bias, especially when peer effects are large relative to the baseline ($\alpha=-1.5$) and when the network is more clustered.}
\label{clustering_bias}
\end{center}
\end{figure}

\begin{figure}[tb]
\begin{center}
\includegraphics[width=\textwidth]{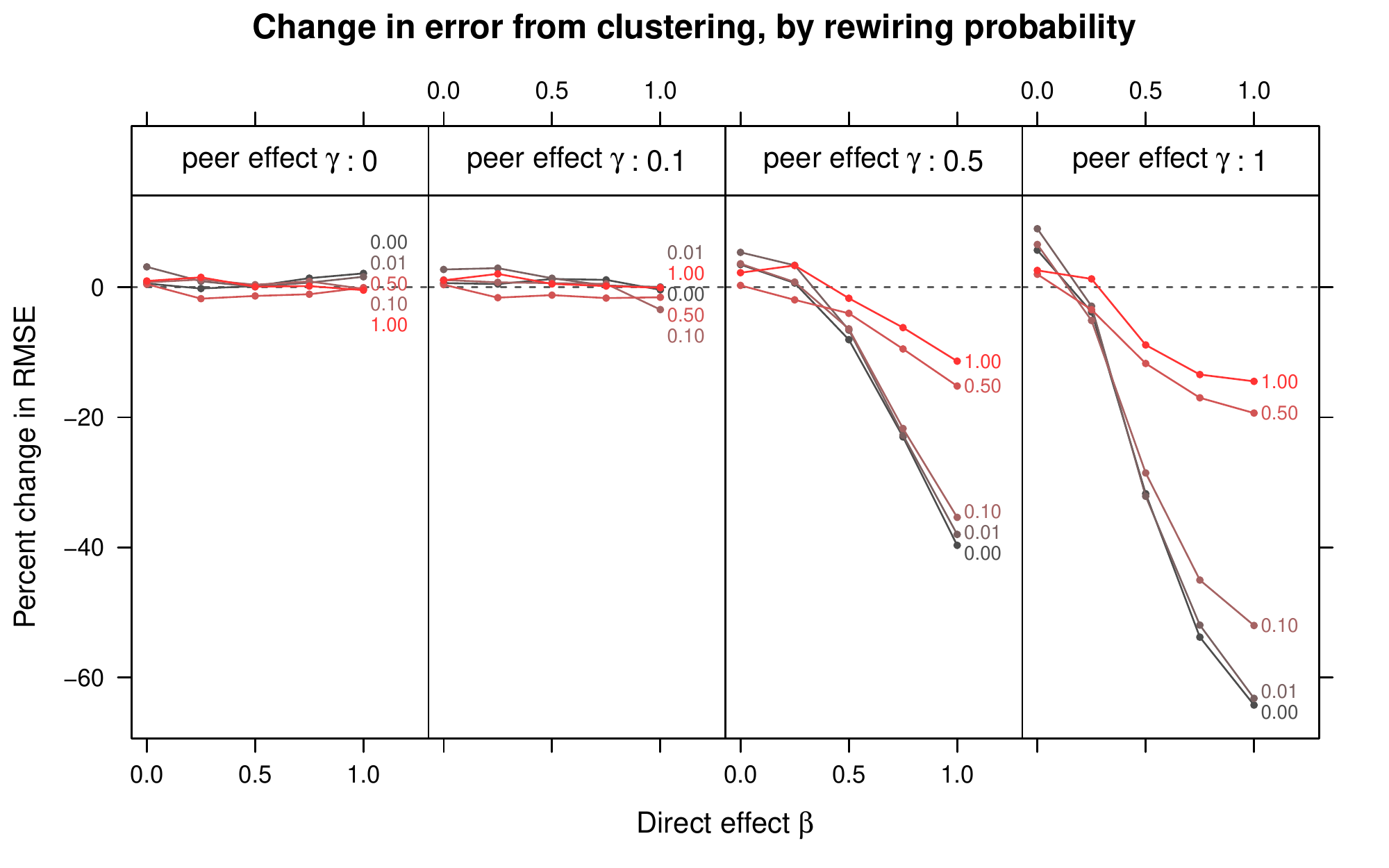}
\caption{Percent change in root-mean-squared-error (RMSE) from clustered assignment for small world networks. While in some cases graph cluster randomization increases RMSE, in other cases (when bias reduction is large), it quite substantially reduces RMSE.}
\label{clustering_rmse}
\end{center}
\end{figure}

We run each of these configurations 5,000 times.
We estimate the true ATE with simulations in which all vertices are put in treatment or control. Each configuration is run 5,000 times for the global treatment case and 5,000 times for the global control case.\footnote{As a variance-reduction strategy for comparisons between designs and true ATE, we use common random numbers throughout the simulations where possible. In particular, for generating observed outcomes, the first instance of each configuration uses the same seed $s_1$, the second instance of each configuration uses the same seed $s_2$, and so on.}

We will now present the results of our simulations of the full process 
of network experimentation. We describe our observations in order to provide insight
into how the different parts of the network experimentation process
interact and contribute to the bias and precision of our experimental estimates.
Our evaluation metrics are bias and root mean squared error (RMSE) 
of the estimated ATE. 

\subsection{Design}
First we examine the bias and mean squared error of the estimated ATE 
for designs using graph cluster randomization compared with 
independent randomization.
In both these cases we use the difference-in-means estimator $\hat{\tau}_{\text{ITR,S}}$.
As expected, using graph cluster randomization reduces bias 
(Figure \ref{clustering_bias}), especially when the peer effects 
and direct effects are large relative to the baseline ($\alpha=-1.5$), and when the network exhibits 
substantial clustering (i.e., the rewiring probability $p_\mathrm{rw}$ is small).

Reduction in bias can come with increases in variance, so it is worth 
evaluating methods that reduce bias also by the effect they 
have on the error of the estimates. 
We compare RMSE, which is increased by both bias and variance, 
between graph cluster randomization and independent 
assignment in Figure \ref{clustering_rmse}. 
In some cases, the reduction in bias comes with a significant 
increase in variance, leading to an RMSE that is either left unchanged 
or even increased. However, in cases where the bias reduction is 
large, this overwhelms the increase in variance, such that graph 
cluster randomization reduces not only bias but also RMSE substantially. 
For example, with substantial clustering ($p_\mathrm{rw} = 0.01$) 
and peer effects ($\gamma = 0.5$), we observe approximately 
40\% RMSE reduction from graph cluster randomization. 
While the RMSE reduction is strongest under substantial clustering,
if both the direct effect strength and peer effect strength are strong,
we observe significant universal reductions in RMSE from 
clustered randomization (though to varied extents), 
regardless of the clustering structure given by $p_\mathrm{rw}$.
It is notable that even with small networks (recall that $N = 1000$),
the bias reduction from graph cluster randomization is large enough
to reduce RMSE.\footnote{Experimenters will generally want to conduct statistical inference, such as through producing standard errors for estimated ATEs, which would need to account
for the increase in variance from graph cluster randomization. We do not treat such methods here.
}

\begin{figure*}[tb]
\begin{center}
\includegraphics[width=\textwidth]{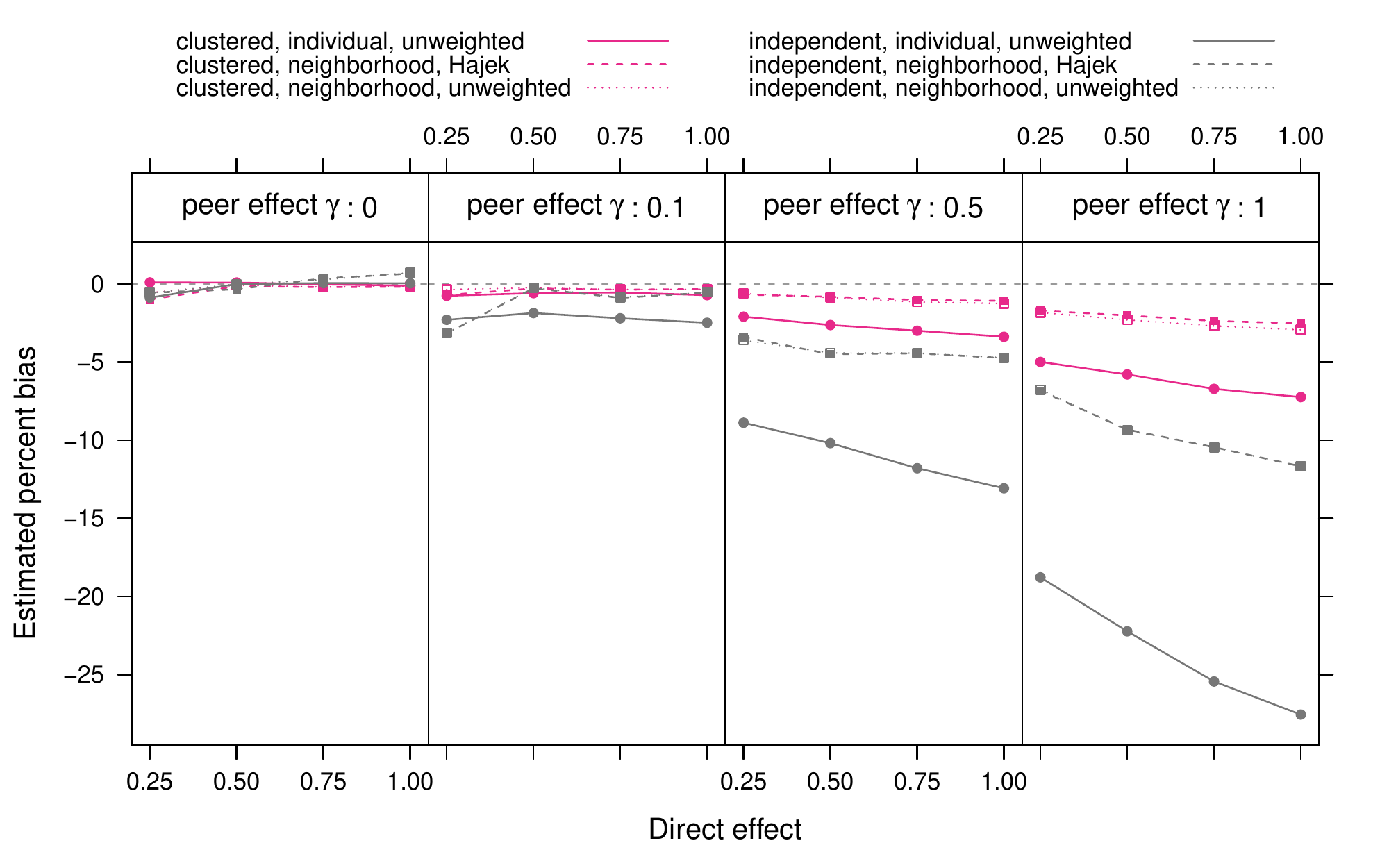}
\caption{Relative bias in ATE estimates for different assignment procedures, exposure models, and estimation methods. The most striking differences are between the assignment procedures, though the neighborhood exposure model also reduces bias (at the cost of increased variance --- see Figure \ref{exposure_model_rmse}). Relative bias is not defined when the true value is zero, so we exclude simulations with the direct effect $\beta = 0$. For all networks, the rewiring probability was $p_\mathrm{rw} = 0.01$.}
\label{exposure_model_bias}
\end{center}
\end{figure*}

\begin{figure*}[tb]
\begin{center}
\includegraphics[width=\textwidth]{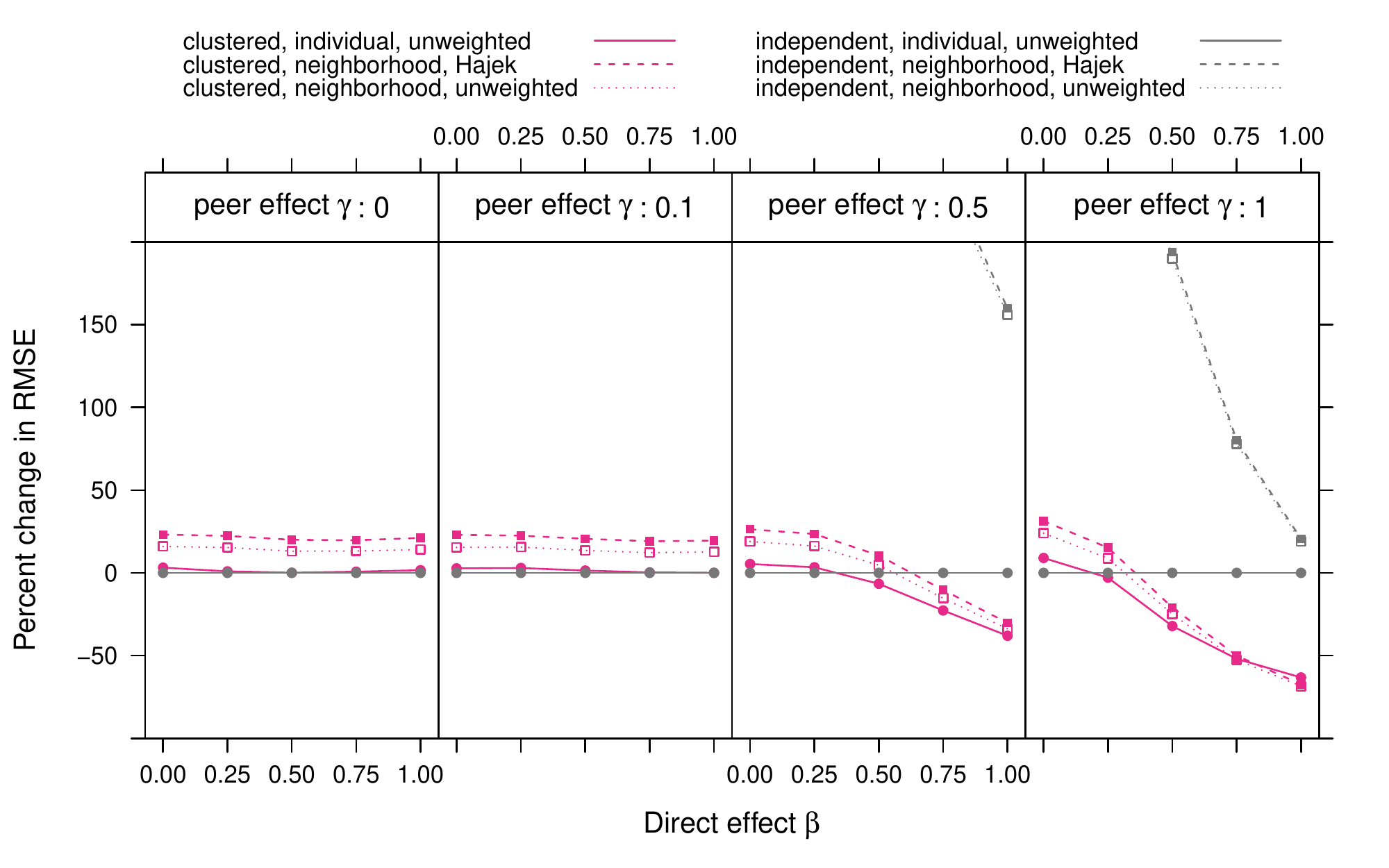}
\caption{Percent change in root-mean-squared-error (RMSE) compared with independent assignment with the simple difference-in-means estimator. Using the neighborhood condition with independent assignment results in large increases in variance: for the two smaller values of $\gamma$, this produces an almost 400\% increase in RMSE. For this reason, the $y$-axis is limited to not show these cases. Rewiring probability $p_\mathrm{rw} = 0.01$.}
\label{exposure_model_rmse}
\end{center}
\end{figure*}

\subsection{Design and analysis}
In addition to changes in design (i.e., graph cluster randomization), we can also use analysis methods intended to account for interference.  We utilize the fractional neighborhood exposure model, which means we only include vertices in the analysis if at least three-quarters of their friends were given the same treatment assignment.\footnote{It is possible for no vertices to meet this condition for treatment or for control. In this case, the estimator is undefined. If this occurs, we expect that experimenters would re-randomize or modify their analysis plan. For the results shown here, we exclude simulations where this occurred, which corresponds to re-randomizing. This did not occur for graph clustered randomization. For independent assignment, this occurred for one of the 5,000 simulations for rewiring probability $p = 0.01$ (i.e., the results shown in Figure~\ref{exposure_model_bias}).}
With this neighborhood exposure model, we consider using propensity score weighting, which corresponds to the Hajek estimator, or ignoring the propensities and using unweighted difference-in-means.  The second estimator has additional bias due to neglecting the propensity-score weights.

Figure \ref{exposure_model_bias} shows several combinations of design 
randomization procedure, exposure model, and estimator.  
We see that using a neighborhood-based definition of effective treatments further 
reduces bias, while the impact of using the Hajek estimator is minimal. 

The low impact of the Hajek estimator follows understandably 
from the fact that small-world graphs do
not exhibit any notable variation in vertex degree, which is the principle
determinant of the propensities used by the Hajek estimator. Thus, for
small-world graphs the weights used by the Hajek estimator are very
close to uniform. With more degree heterogeneity expected in real networks,
the weighting of the Hajek estimator will be more important, 
especially when these heterogeneous propensities are highly correlated with behaviors. 
In general, however, the change in bias from adjusting the analysis 
are not as striking as those from changes due to the experimental design.

Using the neighborhood exposure model means that the estimated average treatment effect
is based on data from fewer vertices, since many vertices may not pass the \emph{a priori} condition.
So the observed modest changes in bias come with increased variance, as reflected in the change in
RMSE compared with independent assignment without using the exposure condition (Figure \ref{exposure_model_rmse}).

\subsection{Results with stochastic blockmodels}

\begin{figure}[tb]
\begin{center}
\subfloat[]{\includegraphics[width=.5\textwidth]{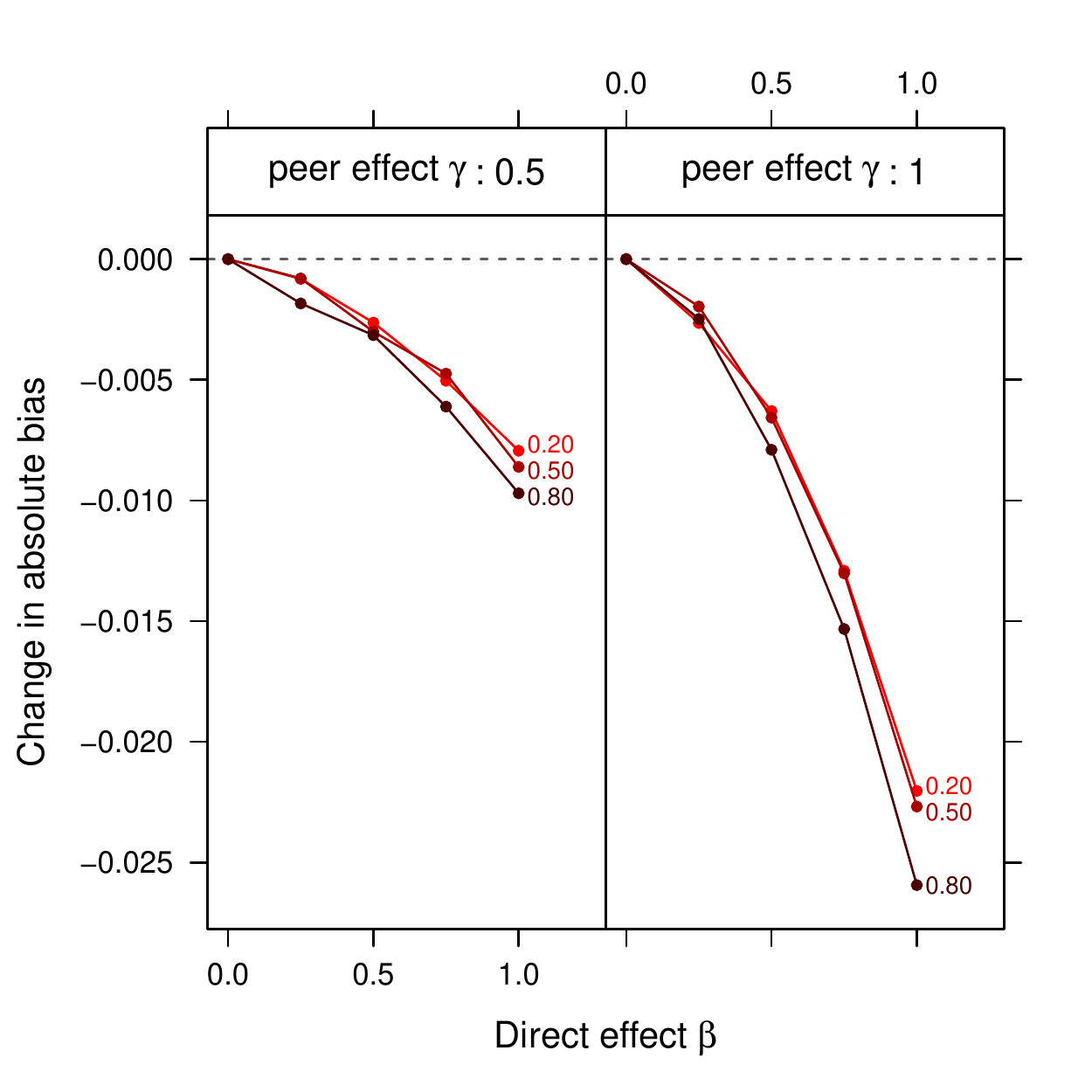}}
\subfloat[]{\includegraphics[width=.5\textwidth]{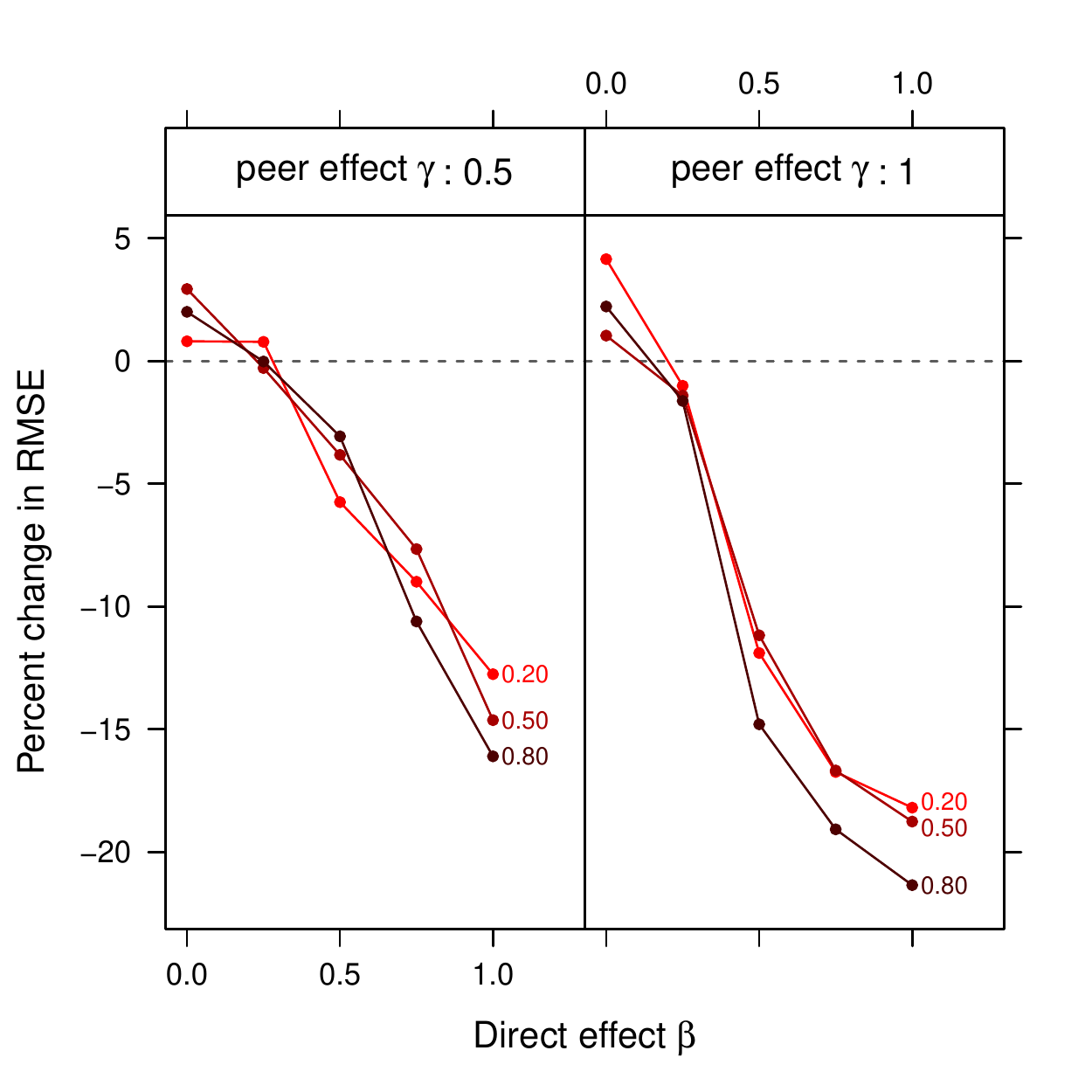}}
\caption{Change in (a) bias and (b) RMSE due to clustered random assignment.
Lines are labeled with the expected proportion of edges that are within a community $p_\text{comm}$.
As before, results vary with the strength of the peer effect $\gamma$,
and the direct effect of the treatment $\beta$.
The largest bias and error reductions here are not as substantial as the largest bias reductions with small-world networks.}
\label{clustering_dcbm}
\end{center}
\end{figure}

\begin{figure*}[tb]
\begin{center}
\includegraphics[width=.7\textwidth]{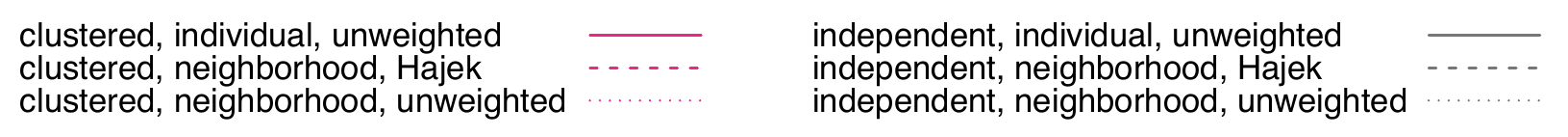} \\ \vspace{-.4cm}
\subfloat[]{\includegraphics[width=.5\textwidth]{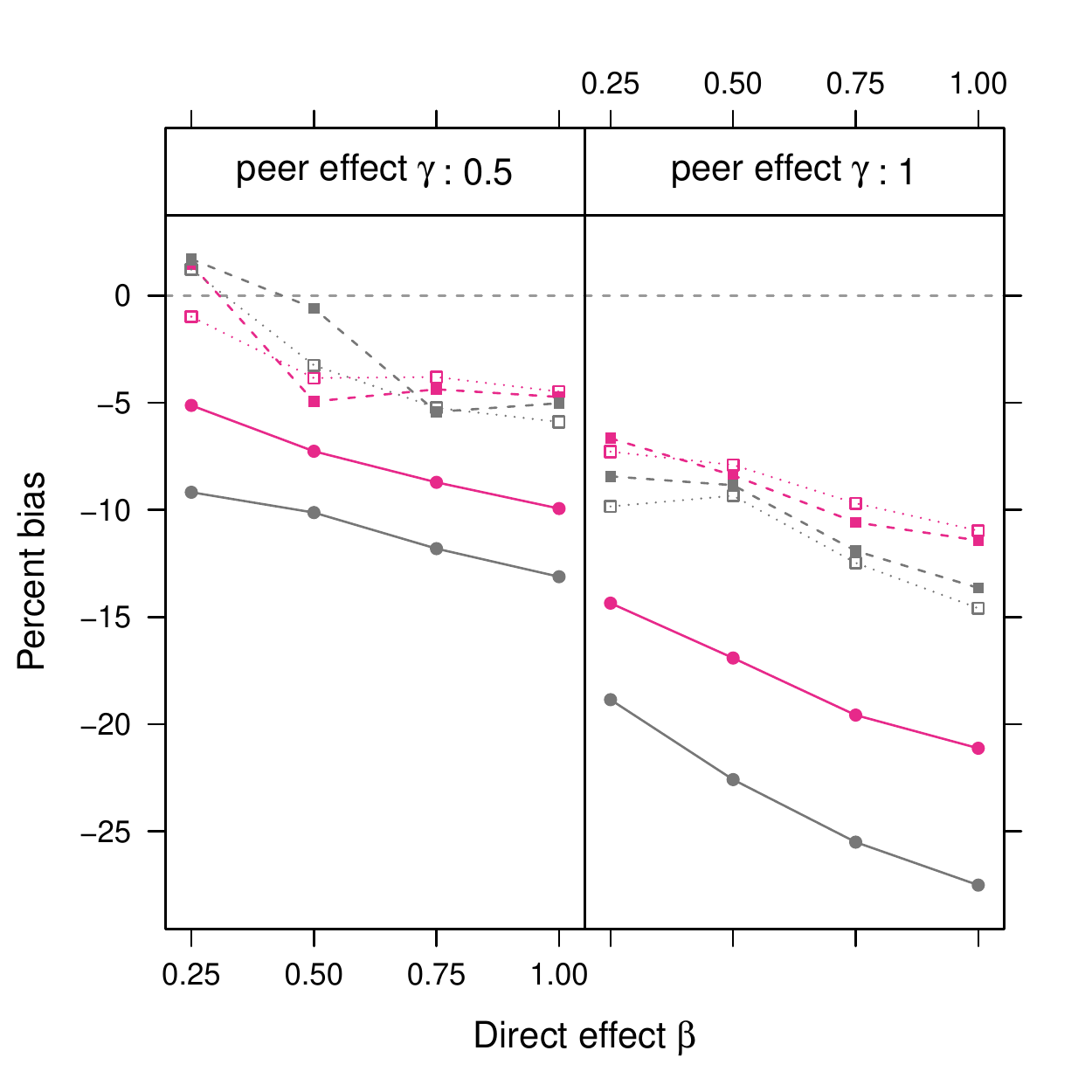}}
\subfloat[]{\includegraphics[width=.5\textwidth]{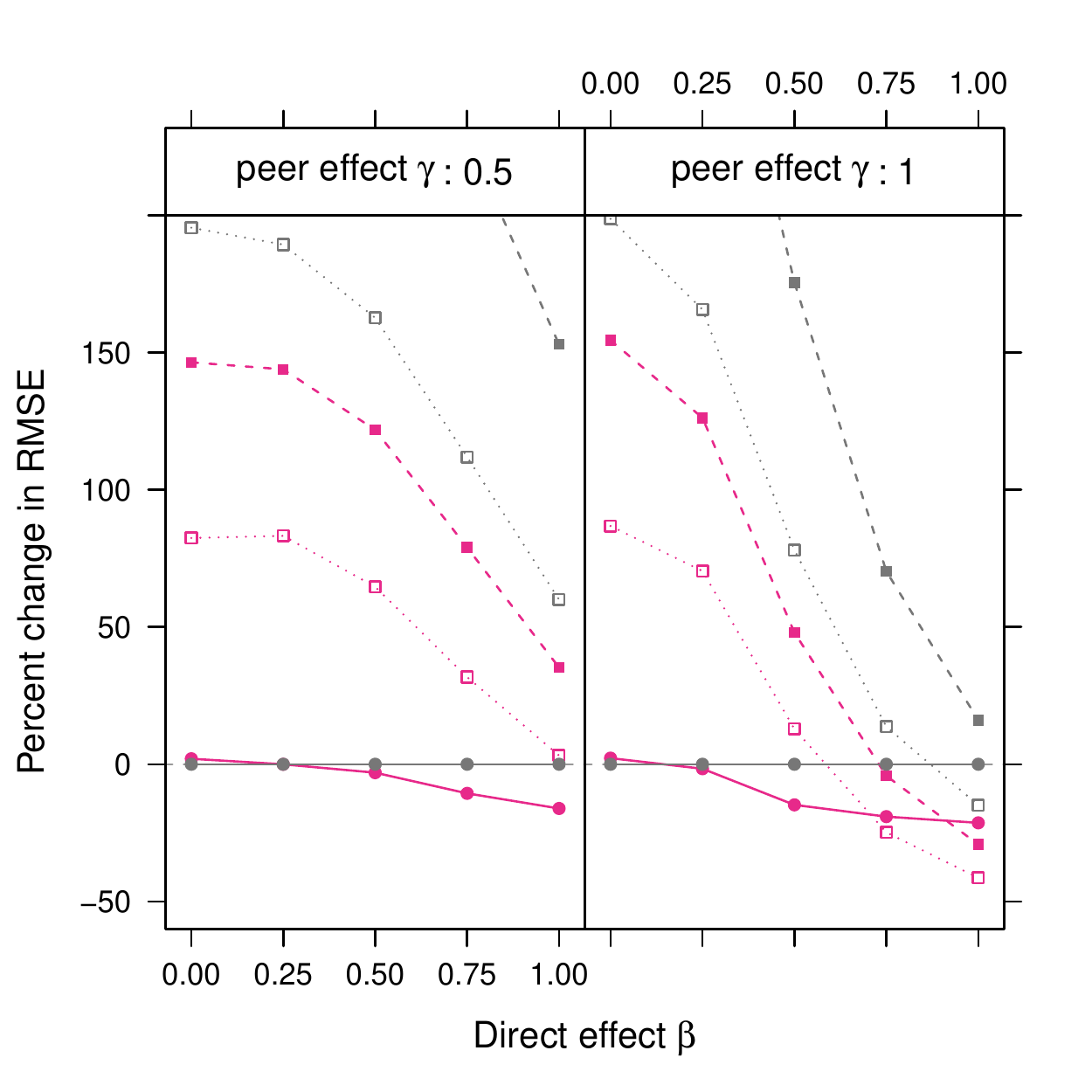}}
\caption{Relative bias (a) and change in RMSE (b) in ATE estimates 
for different assignment procedures, exposure models, and estimation method, 
using the degree-corrected block model with community proportion $p_\text{comm} = 0.8$.
Analysis using the exposure model provides additional bias reduction over using graph cluster randomization only --- with a cost in variance.}
\label{exposure_model_dcbm}
\end{center}
\end{figure*}

As a check on the robustness of these results to the specific choice of network model,
we also conducted simulations with a degree-corrected block model \citep[DCBM;][]{karrer_stochastic_2011},
which provides another way to control the amount of local clustering in a graph and to produce more variation in vertex degree. 

In each simulation, the network is generated according to a degree-corrected block model with 1,000 vertices and 10 communities.
We present results for a subset of the parameter values used with the small-world networks.
Instead of varying the rewiring probability $p_\mathrm{rw}$ to control local clustering, 
we vary the expected proportion of edges that are within a community $p_\text{comm} \in \{0.2, 0.5, 0.8\}$
where vertices are assigned to one of the 10 communities uniformly at random. 
The distribution of expected degrees is a discretized log-normal distribution with mean 10 (as with the small-world networks) and variance 40. 
This produces substantially more variation in degrees than the small-worlds network. Each configuration is repeated 5,000 times.

Figure \ref{clustering_dcbm} displays the change in bias and error that results from graph cluster randomization in these simulations.
The bias and error reduction with the DCBM networks is not as large,
for the same values of other parameters, as with the small world networks. 
We interpret this as a consequence of the presence of higher-degree vertices
and of less local clustering,
even in the simulations with high community proportion (i.e., $p_\text{comm} = 0.8$).\footnote{Note that with $p_\text{comm} = 0.8$
and the chosen degree distribution, the DCBM networks have an average clustering coefficient of approximately 0.095 and average transitivity of approximately 0.091.
This is similar to that of small-world networks with $p_\mathrm{rw} = 0.5$.
This observed bias and error reduction is likewise comparable to that observed with those small-world networks.
}
Qualitative features of these results 
(e.g., bias and error reduction increase with increases in peer effects and increases in clustering)
match those from the small-world networks. 

Figure \ref{exposure_model_dcbm}a displays bias as a function of both design and analysis decisions. As with the small-world networks, estimators making use of the $\lambda$-fractional neighborhood exposure condition reduce bias, whether used with independent or clustered random assignment.
This additional bias reduction comes at the cost of additional variance, such that, in terms of MSE, estimators using the exposure condition are worse for many of the parameter values included in these simulations (Figure \ref{exposure_model_dcbm}b).

\section{Discussion}

Recent work on estimating effects of global treatments in networks through experimentation
has generally started with a particular set of assumptions about patterns of interference,
such as the neighborhood treatment response (NTR) assumption,
that make analysis tractable and then developed estimators with desirable properties (e.g., unbiasedness, consistency) under these assumptions \citep{aronow_estimating_2011,manski_identification_2013}.
Similarly, \citet{ugander_graph_2013} analyzed graph cluster randomization under such assumptions.
Unfortunately, these tractable exposure models are also made implausible by the very processes, such as peer effects, 
that are expected to produce interference in the first place. 
Therefore, we have considered what can be done about bias from interference
when such restrictions on interference cannot be assumed to apply in reality.

The theoretical analysis in this paper offers sufficient, but not necessary, conditions for this bias reduction through design and analysis
in the presence of potentially global interference.
To further evaluate how design and analysis decisions can reduce bias, we reported results from simulation studies 
in which outcomes are produced by a dynamic model that includes peer effects.
These results suggest that when networks exhibit substantial clustering
and there are both substantial direct and indirect (via peer effects) effects of a treatment,
graph cluster randomization can substantially reduce bias with comparatively small increases in variance.
Significant error reduction occurred with networks of only 1,000 vertices,
highlighting the applicability of these results beyond experiments on large networks.
Additional reductions in bias can be achieved through the specific estimators used,
even though these estimators are based on incorrect assumptions about effective treatments.

Further work should examine how these results apply to other networks and data-generating processes.
The theoretical analysis and simulations in this paper used models in which outcomes are monotonic in treatment and peer behavior.
Such models are a natural choice given many substantive theories,
but in other cases vertices will be expected to be less likely to take an action as more peers take that action.
Our simulations did not include vertices characteristics (besides degree) and prior behaviors,
which could play an important role in the bias and variance for different designs and estimators.
Much of the empirical literature that considers peer effects in networks,
whether field experiments \citep[e.g.,][]{aral_creating_2011,bakshy_social_2012,bapna_paid_2012}
or observational studies \citep[e.g.,][]{aral_distinguishing_2009,goldsmith_social_2013}
has aimed to estimate peer effects themselves,
rather than estimating effects of interventions that work partially through peer effects;
a fruitful direction for future work would involve directly modeling the peer effects involved and 
then using these models to estimate effects of global treatments \citep[cf.][]{laan_causal_2014}.\footnote{In the language of contemporary econometrics, one could describe the present work as taking a ``reduced form'' approach to this problem, rather than
trying to learn about the underlying data generating process through estimating ``structural'' parameters.}
This could substantially expand the range of designs and analysis methods to consider.

\subsubsection*{Acknowledgements}
We are grateful for comments from Edoardo Airoldi, Eytan Bakshy, Thomas Barrios, 
Daniel Merl, Cyrus Samii,
and participants in 
the Statistical and Machine Learning Approaches to Network Experimentation Workshop at Carnegie Mellon University,
and in seminars in 
the Stanford Graduate School of Business,
Columbia University's Department of Statistics,
New York University's Department of Politics,
and
UC Davis's Department of Statistics.

\bibliographystyle{imsart-nameyear}

\appendix
\section{Appendix}

\subsection{Modified graph cluster randomization: hole punching}
\label{appendix:hole_punching}
We now briefly present a simple modification of graph cluster randomization that adds vertex-level randomness to the treatment assignment, such that some vertex assignments may not match their cluster assignment.
We set
$$
W_i \sim \text{Bernoulli}(q_{C(i)})
$$
$$
X_i \sim \text{Bernoulli}(\eta)
$$
$$
Z_i = X_i W_{C(i)} +  (1 - X_i)(1 - W_{C(i)}).
$$
The $X_i$ are independent switching variables that set $Z_i$ to $W_{C(i)}$ with probability $\eta$, typically high, and flip the assignment otherwise (``punch a hole'').
That is, clusters are assigned to have their vertices predominantly in one of treatment or control. We call this modification \emph{hole punching}, because it inverts the treatment condition of a small fraction of vertices, placing them in a highly isolated treatment position within their cluster. This modification could be useful for estimating differences between direct and peer effects, since it results in many vertices experiencing the direct treatment without peer effects or the peer effects without the direct treatment.
It also has the appealing consequence of avoiding exact zero probabilities of assignment to some vectors $Z$.
This is important in cases where one might want to compare outcomes as a function of number of peers assigned to the treatment;
otherwise, many of these comparisons would be between conditions that many vertices could not be assigned to.

\subsection{Bias reduction from design: balanced linear case}
\label{appendix:balanced_linear_case}

In this appendix, we consider the linear outcome model under an alternative graph cluster randomization that enforces balance
(i.e., equal sample sizes in treatment and control)
Assume there is an even number of clusters $N_C$, each with $N / N_C$ vertices.
Pick $N_C / 2$ clusters at random and assign them to treatment; assign the remaining clusters to control.

\begin{theorem}
\label{prop:linear_bias_reduction_design_balanced}
Assume we have a linear outcome model for all vertices $i \in V$ such that
\begin{eqnarray}
\E_U[Y_i(z, U)] = a_i + \sum_{j \in V} B_{ij} z_j
\end{eqnarray}
and further assume that $Y_i(z, u)$ is monotonically increasing in $z$
for every $u \in \mathbb{U}^N$ and 
vertex $i$
such that $B_{ij} \geq 0$.

Then for some mapping of vertices to clusters,
the absolute bias of $\tau^d_\mathrm{ITR}(1,0)$ when $d$ is graph cluster randomization is less than or equal to 
the absolute bias when $d$ is independent assignment, with a fixed treatment probability $p$.
\end{theorem}

\begin{proof}
Using the linear model for $Y_i$ and the definition of $\tau$, we have that the true ATE $\tau$ is given by
\begin{align}
\label{eq:linear_outcome_model}
\tau(1,0) = \mu(1) - \mu(0) = \frac{1}{N} \sum_{ij} B_{ij}
\end{align}
for this outcome model.  Under balanced graph cluster randomization,
\begin{align}
\tau^\mathrm{bgcr}_\mathrm{ITR}(1,0) &= \frac{1}{N} \sum_{ij} B_{ij} \left[
\mathbf{1}[C(i) = C(j)]  + \mathbf{1}[C(i) \neq C(j)] \left(\frac{N_C/2 - 1}{N_C -1} - \frac{N_c / 2}{N_c - 1} \right)
\right] \\
&= \frac{1}{N} \sum_{ij} B_{ij} \left[
\mathbf{1}[C(i) = C(j)]  -  \frac{\mathbf{1}[C(i) \neq C(j)]}{N_C -1} \right].
\end{align}
We can extend this to the case where the mapping of vertices to clusters is random:
\begin{align}
\tau^\mathrm{bgcr}_\mathrm{ITR}(1,0) &= \frac{1}{N} \sum_{ij} B_{ij} \left[
\Pr(C(i) = C(j))  -  \frac{\Pr(C(i) \neq C(j))}{N_C -1} \right].
\end{align}
Separating out $B_{ii}$:
\begin{align}
\tau^\mathrm{bgcr}_\mathrm{ITR}(1,0) &= \frac{1}{N} \left( \sum_i B_{ii} + \sum_{ij; j \neq i} B_{ij} \left(
\Pr(C(i) = C(j))  -  \frac{\Pr(C(i) \neq C(j))}{N_C -1} \right) \right).
\end{align}
If we have uniform probability over all cluster assignments with the same number of vertices per cluster, then for $i \neq j$,
$$
\Pr(C(i) = C(j)) = \frac{N / N_C - 1}{N},
$$
so
\begin{align}
\tau^\mathrm{bgcr}_\mathrm{ITR}(1,0) &= \frac{1}{N} \left( \sum_i B_{ii} - \sum_{ij; j \neq i} B_{ij} \frac{N_C}{(N_C - 1)N} \right).
\end{align}
Under balanced independent assignment, we just have $N_C = N$, so
\begin{align}
\tau^\mathrm{bind}_\mathrm{ITR}(1,0) &= \frac{1}{N} \left( \sum_i B_{ii} - \sum_{ij; j \neq i} B_{ij} / (N - 1) \right).
\end{align}
Because $B_{ij} \geq 0$, together this implies that
$\tau(1,0) - \tau^\mathrm{gcr}_\mathrm{ITR}(1,0) \leq \tau(1,0)-\tau^\mathrm{ind}_\mathrm{ITR}(1,0)$, where monotonicity again dictates that each side of this inequality is positive. 
\end{proof}

The proof showed that clustering can reduce bias over independent assignment when preserving balance.  The 
relative bias for graph cluster randomization that preserves balance is
\begin{eqnarray}
\tau^\mathrm{gcr}_\mathrm{ITR}(1,0)/\tau(1,0) - 1 &=& \frac{\sum_{ij} B_{ij} \left[\mathbf{1}[C(i) = C(j)]  -  \frac{\mathbf{1}[C(i) \neq C(j)]}{N_C -1}\right]}{\sum_{ij} B_{ij}} - 1 \\
&=& \left(1 + \frac{1}{N_C - 1}\right)\left(\frac{\sum_{ij} B_{ij} \mathbf{1}[C(i) = C(j)]}{\sum_{ij} B_{ij}} - 1\right). \nonumber
\end{eqnarray}
which is the same expression as the relative bias for graph cluster randomization except for the multiplicative factor in the front.  For large enough
$N_C$, the relative biases will be identical, and therefore meaningful relative bias reduction occurs depending only on the clustering's relationship to the values
$B_{ij}$, and not whether the sampling scheme preserves balance or not.

\subsection{Bias reduction from analysis}
\label{appendix:theory_analysis}
Here we restate and prove Theorem \ref{prop:bias_reduction_analysis} from the main text.
We also consider two possible extensions of this theorem to graph cluster randomization (from independent random assignment), giving a counterexample for one extension and proving an analog of the theorem for the other extension.

Consider functions $g_i(\cdot)$ such that $g_i(Z) = g_i(z)$ just implies that 
for some subset of vertices $J_i$ we have that 
$\sum_{j \in J_i} 1\{Z_j = z_j\} \geq l_i$ and that $Z_i = z_i$. 
These are conditions such that some subset of size $l_i$ of a set of vertices 
has treatment assignment matching that in the global treatment vector of interest $z$.
The ITR and NTR assumptions both are of this type,
where with ITR $J_i$ is the empty set 
and with NTR $J_i = \delta(i)$ and $l_i=k_i$, $i$'s degree. 
The fractional relaxation of NTR (FNTR) is also of this type, 
with $J_i = \delta(i)$ and $l_i = \lceil \lambda k_i \rceil$.

If we have two such functions $g_i^A(\cdot)$ and $g_i^B(\cdot)$ with the same $J_i$, and $g_i^A(z) = g_i^A(z')$ implies $g_i^B(z) = g_i^B(z')$, we say that $g_i^A(\cdot)$ is more restrictive than $g_i^B(\cdot)$.

\begin{reptheorem}{prop:bias_reduction_analysis}
Let $g^A(\cdot)$ and $g^B(\cdot)$ be vectors of such functions where $g_i^A(\cdot)$ is more restrictive than $g_i^B(\cdot)$ for every vertex $i$, and let independent random assignment be the experimental design.  A sufficient condition for estimand $\tau^\mathrm{ind}_{g^A}(1, 0)$ to have less than or equal absolute bias than $\tau^\mathrm{ind}_{g^B}(1, 0)$,
where these estimands are defined by Equation \ref{effective_treatment_tau}, is that we have monotonically increasing responses or 
monotonically decreasing responses for every vertex with respect to $z$.
\end{reptheorem}

\begin{proof}
All expectations are taken with respect to
independent random assignment.  Assume monotonically increasing responses for every vertex
and select an arbitrary vertex $i$.  Let
\begin{eqnarray}
\tilde{Y_i}(z_{J_i}) =  \E_{Z_{V/{J_i}}}[Y_i(z_i = 1, Z_{V/{J_i}}, z_{J_i})].
\end{eqnarray}
This quantity is the expectation of the potential outcome for $i$ when $z_i = 1$ and the subset of $z$ corresponding
to $J_i$ is set to $z_{J_i}$.  The monotonicity of $Y_i$ carries over to $\tilde{Y_i}(z_{J_i})$.   

To reduce the notation in what follows, we define $A_i$ to be the event that $g_i^A(Z) = g_i^A(1)$ and
$B_i$ to be the event that $g_i^B(Z) = g_i^B(1)$.  We also define $q_i(Z) = \sum_{j \in J_i} 1\{Z_j = 1\}$.  Then
\begin{eqnarray}
\E[\tilde{Y}_i | A_i ] &=& \sum_{q \geq l_i^A}^{|J_i|} \E[\tilde{Y}_i | q_i(Z) = q] P(q_i(Z) = q | A_i), \nonumber\\
\E[\tilde{Y}_i | \neg A_i \land B_i ] &=& \sum_{q \geq l_i^B}^{l_i^A - 1} \E[\tilde{Y}_i | q_i(Z) = q] P(q_i(Z) = q | \neg A_i \land B_i ).
\label{eq:conditional_exp}
\end{eqnarray}

Due to independent random assignment, conditioning on $q_i(Z) = q$ means uniformly sampling
a $z_{J_i}$ that has $q$ ones and $|J_i| - q$ zeroes.  Consider the following process where $q < |J_i|$.  Randomly select
a $z_{J_i}$ with $q$ ones and $|J_i| - q$ zeroes.  Select at random a $0$ element and change it into a $1$ to create
another vector $z'_{J_i}$.  Record both $\tilde{Y}_i(z_{J_i})$ and $\tilde{Y}_i(z'_{J_i})$ as a pair of values.  Due to the monotonicity of
$\tilde{Y}_i$, we have that $\tilde{Y}_i(z_{J_i}) \leq \tilde{Y}_i(z'_{J_i})$.

In this process, $z_{J_i}$ is a uniformly sampled vector that has $q$ ones and $|J_i| - q$ zeroes, and $z'_{J_i}$ is a uniformly sampled vector that 
has $q+1$ ones and $|J_i| - (q+1)$ zeroes. Repeating this process an infinite number of times and using the 
empirical average of the $\tilde{Y}_i(z_{J_i})$'s computes $\E[\tilde{Y}_i | q_i(Z) = q]$.  Similarly,
the empirical average of the $\tilde{Y}_i(z'_{J_i})$ computes $\E[\tilde{Y}_i | q_i(Z) = q+1]$.  Due to the per sample inequality, this shows
that $\E[\tilde{Y}_i | q_i(Z) = q] \leq \E[\tilde{Y}_i | q_i(Z) = q+1]$.  By induction, $\E[\tilde{Y}_i | q_i(Z) = q] \leq \E[\tilde{Y}_i | q_i(Z) = q']$ when $q < q'$.  
Combining this with Eq.~\ref{eq:conditional_exp},

\begin{eqnarray}
\E[\tilde{Y}_i | \neg A_i \land B_i ] \leq \E[\tilde{Y}_i | A_i ].
\label{eq:conditional_inequality}
\end{eqnarray}

Since the design is independent random assignment, we have that
\begin{eqnarray}
\E[Y_i | B_i] &=& \E[\tilde{Y}_i | B_i] \nonumber\\
&=& \E[\tilde{Y}_i | A_i] P(A_i | B_i) + E[\tilde{Y}_i | \neg A_i \land B_i] P(\neg A_i | B_i).
\end{eqnarray}
where in the second equality we have used that $g_i^A$ is more restrictive than $g_i^B$ and that the set $J_i$ is common to
both $g_i^A$ and $g_i^B$.  With Eq.~\ref{eq:conditional_inequality}, this implies
\begin{eqnarray}
\E[Y_i | B_i] \leq  \E[\tilde{Y}_i | A_i] = \E[Y_i | A_i].
\end{eqnarray}

Since this inequality applies for all vertices $i$, we therefore have that
\begin{eqnarray}
\mu^\mathrm{ind}_{g^B}(1) \leq \mu^\mathrm{ind}_{g^A}(1),
\end{eqnarray}
from which we immediately conclude that $g^A$ has less absolute bias for $\mu(1)$ than $g^B$.  An analogous 
argument applies for $\mu(0)$, proving that $\tau^\mathrm{ind}_{g^A}$ has less absolute bias for
$\tau(1,0)$, the average treatment effect.  

The proof for monotonically decreasing responses follows when switching the inequalities throughout
the above.
\end{proof}

This proposition demonstrates how using more restrictive exposure conditions can be helpful in reducing bias, but the
proposition just applies to independent assignment, rather than graph cluster randomization. To show why it does not hold for graph cluster randomization, we present the following
counterexample with two fractional neighborhood treatment response (FNTR) effective treatments.

Consider some vertex $i$ with no neighbors in its own cluster, and three other clusters present in its neighborhood:
one cluster with 10 neighbors, one cluster with one neighbor, and another cluster with one neighbor;
call this last neighbor vertex $a$.
Let $Y_i = 1$ when $Z_a = 1$ and $Z_i = 1$, and let $Y_i = 0$ otherwise.
Let the less restrictive function $g^B_i(\cdot)$ require that at least 2 neighbors match the global treatment vector, and let the more restrictive function $g^B_i(\cdot)$ require that at least 3 neighbors match; that is, let $l^B_i = 2$ and $l^A_i = 3$. Then under graph cluster randomization, we have $\E[Y_i \vert A_i] \approx 0.5$, but $\E[Y_i \vert B_i] \approx 0.6$. So using the more restrictive function actually increases bias in this somewhat extreme scenario.

While this counterexample demonstrates that using more restrictive exposure conditions of this kind is not always helpful under
graph cluster randomization, we do observe bias reduction in our simulations using graph cluster randomization
without meeting the sufficient conditions of the theorem.  In general, we expect that for bias to increase, there must be heterogeneous
effects across heterogeneously sized clusters as in the counterexample above.

In fact, with a redefinition of the exposure conditions, we can provide a similar proposition that does include graph cluster randomization and also
encompasses independent assignment as a special case.

\begin{corollary}
Consider a fixed set of clusters which will be used for graph cluster randomization.  Let function $g_i(\cdot)$, 
for all vertices $i$, be such that $g_i(Z) = g_i(z)$ implies that some subset of clusters $J_i$ which do not include $i$ we have
that $\sum_{C \in J_i} 1\{Z_C = z_C\} \geq l_i$ (at least $l_i$ of the clusters in $J_i$ match the global treatment vector $z$ exactly),
and $Z_i = z_i$.  Consider two such functions where $g^A_i(\cdot)$ is more 
restrictive than $g^B_i(\cdot)$ for all $i$.  Then a sufficient condition for estimand $\tau^\mathrm{gcr}_{g^A}(1, 0)$ to have less than or equal absolute bias 
than $\tau^\mathrm{gcr}_{g^B}(1, 0)$, where these estimands are defined by Equation \ref{effective_treatment_tau}, is that we have monotonically 
increasing responses or monotonically decreasing responses for every vertex with respect to $z$. 
\end{corollary}

\begin{proof}
This proof is essentially the same as for Theorem \ref{prop:bias_reduction_analysis} except $\tilde{Y}_i$ is redefined as
\begin{eqnarray}
\tilde{Y_i}(z_{J_i}) =  \E_{Z_{V/{J_i}}}[Y_i(z_{C_i} = 1, Z_{V/{J_i}}, z_{J_i})],
\end{eqnarray}
expectations are computed with respect to graph cluster randomization instead of independent treatment assignment, and references
to 1's and 0's apply to clusters in $J_i$.
\end{proof}

An important special case of this corollary covers the comparison of FNTR with ITR under graph cluster randomization, since FNTR and ITR can be
written as cluster-level exposure conditions of this kind.

\end{document}